\documentclass[12pt,anon]{l4dc2024}
\DeclareMathOperator*{\argmin}{arg\,min}
\title[]{Distributed Learning and Function Fusion in Reproducing Kernel Hilbert Space}
\usepackage{algorithm}
\usepackage{algorithmic}
\usepackage{wrapfig}
\usepackage[mathscr]{euscript}
\usepackage{times}
\allowdisplaybreaks
\newtheorem*{assumption}{Assumption}
\author{%
 \Name{Aneesh Raghavan} \Email{aneesh@kth.se}\\
 \addr DCS Division, KTH, Royal Institute of Technology, Stockholm
 \AND
 \Name{Karl {Henrik Johansson}} \Email{kallej@kth.se}\\
 \addr DCS Division, KTH, Royal Institute of Technology, Stockholm
}
\begin{document}

\maketitle

\begin{abstract}
We consider the problem of function estimation by a multi-agent system comprising of two agents and a fusion center. Each agent receives data comprising of samples of an independent variable (input) and the corresponding values of the dependent variable (output). The data remains local and is not shared with other members in the system. The objective of the system is to collaboratively estimate the function from the input to the output. To this end, we develop an iterative distributed algorithm for this function estimation problem. Each agent solves a local estimation problem in a Reproducing Kernel Hilbert Space (RKHS) and uploads the function to the fusion center. At the fusion center, the functions are fused by first estimating the data points that would have generated the uploaded functions and then subsequently solving a least squares estimation problem using the estimated data from both functions. The fused function is downloaded by the agents and is subsequently used for estimation at the next iteration along with incoming data. This procedure is executed sequentially and stopped when the difference between consecutively estimated functions becomes small enough. To analyze the algorithm, we define learning operators for the agents, fusion  center and the system. We study the asymptotic properties of the norm of the learning operators and find sufficient conditions under which they converge to $1$. Given a sequence of data points, we define and prove the existence of the learning operator for the system. We prove that the porposed learning algorithm is consistent and demonstrate the same using an example.  
\end{abstract}


\section{Introduction}\label{Section 1}
Distributed learning algorithms develop models where multiple instances of the same model are trained using different subsets of the training data set or parallel paths of a single model are trained at multiple nodes using the same or different data sets. Algorithms often focus on  parallelizing computing for efficient learning, e.g.,  \cite{chan1993toward}, \cite{verbraeken2020survey}, \cite{peteiro2013survey}. For the fusion of the models, some algorithms use tools from meta-learning. Federated learning has emerged as an efficient approach for distributed learning using heterogeneous data. Has found applications in IoT, healthcare etc, e.g. \cite{liu2022distributed}, \cite{nguyen2021federated}.

Multimodal learning algorithms use data sets obtained using multiple kinds of sensors for training models. For example, images and 3D depth scans can be used for edge detection,and, audio and visual data for speech recognition, see, \cite{baltruvsaitis2018multimodal}, \cite{ngiam2011multimodal}. \cite{lanckriet2004learning} study the problem of learning a kernel matrix using data, from the space of kernel matrices generated by linear combinations of known kernel matrices. Following the same idea, multiple kernel learning (MKL) algorithms have been investigated. Multimodal learning using kernel methods has been applied to disease detection \citep{duan2012domain}, \citep{liu2013multiple},  sentiment analysis \citep{poria2017ensemble}, emotion recognition \citep{sikka2013multiple}, etc. Most of these algorithms are centralized, i.e., data from sensors are collected and analyzed simultaneously at one location.

Formal approaches to study knowledge and it is properties are crucial for development of intelligent collaborative multi-agent systems \citep{rosenschein1985formal}. (Dynamic) Epistemic logic is used as a formal language to describe knowledge and learning formally \citep{van2007dynamic}. Knowledge can be viewed as mapping from the set of events  to the set $\{0,1, \varnothing\}$, i.e., if an event is true, false or its validity is unknown. Given input-output data, the function estimated from it  enables us to state which events of the form ``input = $x$ and output = $y$" are true and which are false as long as $(x,y)$ belongs to the domain of the estimated function. In this spirit, we use the term \textit{Knowledge} for the mapping learned and the term \textit{Knowledge Space} for the function space where the agent is learning. 

The problem considered in this paper is as follows. There are two agents, Agent 1 and Agent 2, receiving data comprising of an independent variable and the corresponding value of a dependent variable. The data is received sequentially, one sample at a time. There exists a fusion center that can communicate with each agent. When an agent transmits to a fusion center, it is referred to as \textit{upload operation}. When the fusion center transmits to the an agent, it is referred to as \textit{download operation}. The data collected by the agents is private to the agents and is not shared with other agents in the system. Our objective is to develop an iterative collaborative function estimation scheme which eventually converges to an estimate of the mapping from the input to the output.

For motivation, we consider the following example from distributed SLAM \cite{chellali2013distributed}, \cite{tian2022kimera}, \cite{lajoie2020door}. There are two or more agents in an environment with different onboard sensors whose individual aim during the learning phase  is to find a mapping from their true position  to the sensor output. During the execution phase, this mapping could be utilized for planning and completion of tasks. During the learning phase, each agent is restricted to survey a certain region of the environment they live in and are aware of the predominant features of the map from the position to the sensor output. By collecting data, they each estimate a map from position to sensor output. These maps can be viewed as partial knowledge of the environment. The main objective of this paper is to come with an algorithm which can fuse local knowledge at agents to obtain global knowledge of the system. 

Our contributions are as follows. We develop a learning scheme for the problem. Given the downloaded function at stage $n-1$ and data point at stage $n$, each agent estimates the mapping from the input to the output by solving a least-squares regression problem. The estimated functions are uploaded to the fusion center. At the fusion center, the data received by the agents is estimated from the functions received considering that the agents would have performed optimal estimation. Using the estimated data points, a least-squares regression problem is solved to obtain the fused function. The fused function is downloaded on to the knowledge space of each agent. $n$ is incremented by 1 and the sequence is repeated. We define learning operators for each member of the system and for the system itself. We present sufficient conditions under which the norm of the operators converge to $1$. We prove that the sequence of learning operators at every iteration is uniformly bounded. We prove the consistency of the learning scheme and the existence of fixed point, i.e., a set of initial estimates for which given a data sequence the eventually learned function is the same as the initial estimates. We present an example demonstrating the algorithm. Note that, the data collected by the agents is transformed into the estimated functions and forgotten. At the fusion center, when the data the could have generated these functions are estimated it is relearned. Our previous work, \cite{raghavan2023distributed}, is a special case of the work presented here as in the former, we considered learning in the same knowledge space for the both agents, the  estimation problem was a one shot problem and the fusion problem was an optimization problem over linear combinations of the functions estimated by the agents. 

The organization of the paper is as follows. In Section \ref{Section 2}, we describe the learning architecture, the estimation problems and the learning algorithm. In Section \ref{Section 3}, we present the main results including the solution to the estimation problems and the consistency of the learning algorithm. In Section \ref{Section 4}, we present an example demonstrating the learning algorithm. We conclude with some comments and future work in Section \ref{Section 5}. Notation: for a function $f\in V$, $V$ vector space, we use the notation $f$ when it is treated as a vector and the notation $f(\cdot)$ when it is treated as a function. The null element of a vector space, $V$, is denoted by $\theta_{V}$. The span of vectors $\{v_{j}\}^{j=n}_{j=1} \subset V$ is denoted as $\text{Span}\Big(\{v_{j}\}^{j=n}_{j=1}\Big)$. The projection onto a subspace $\mathcal{M}$ of a Hilbert space $H$ is denoted by  $\Pi_{\mathcal{M}}$. The dual space of a vector space is denoted by $V^*$.
\section{The Distributed Learning System}\label{Section 2}
In this section, we describe the learning system. We begin with the description of the learning architecture, followed by the estimation and fusion problems, and finally present the learning algorithm.
\subsection{The Learning Architecture}\label{Subsection 2.1}
\begin{figure}
\begin{center}
\includegraphics[scale=0.45]{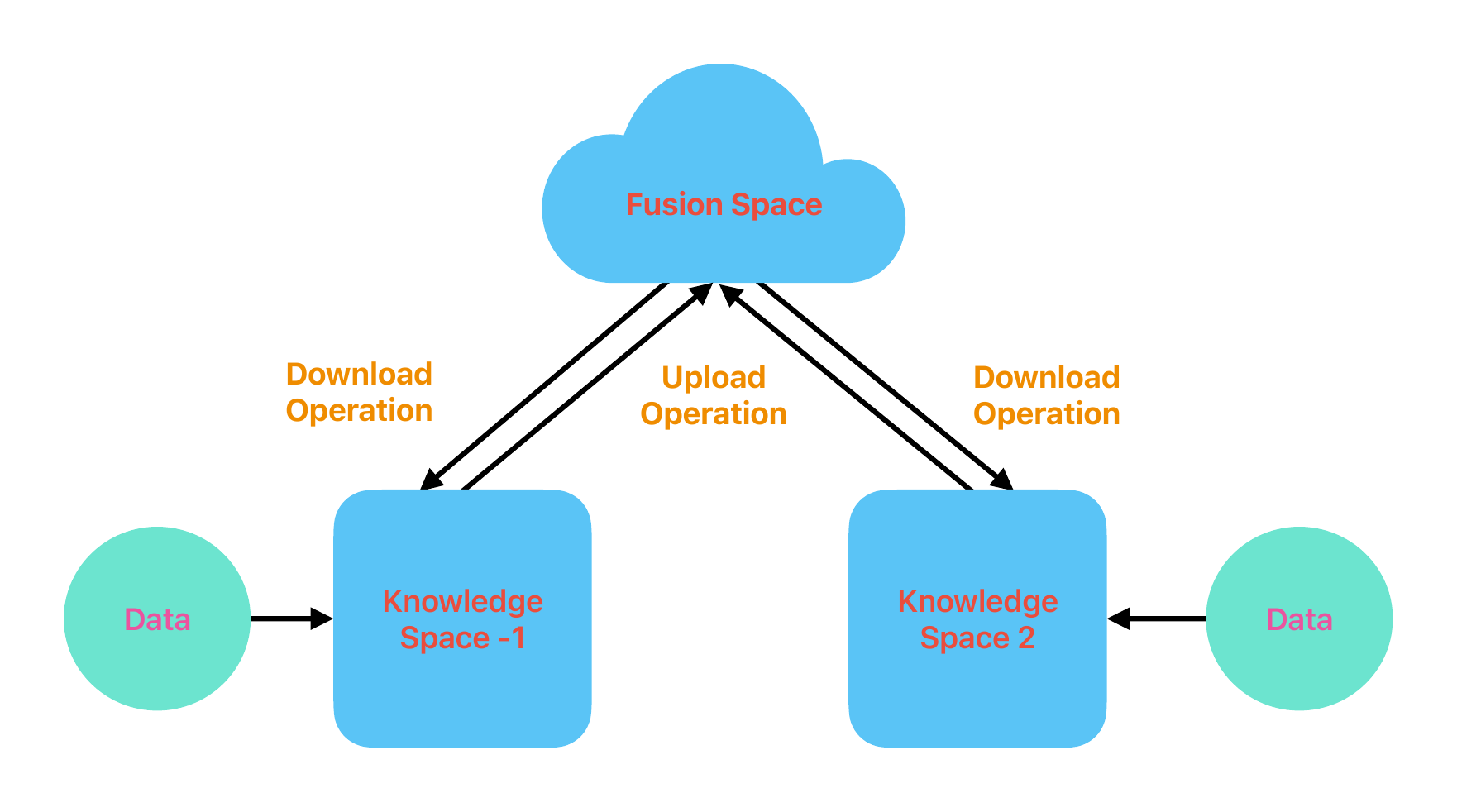}
\caption{The Learning Architecture}\label{Figure 1}
\end{center}
\end{figure}
The learning system comprises of two agents, Agent $1$ and Agent $2$, and a fusion center.  Let $\mathcal{X} \subset \mathbb{R}^{I}$. The set of features for agent $i$ is a set of continuous functions, $\{\varphi^{i}_{j}(\cdot)\}_{j \in \mathcal{I}^i}$, where $\varphi^{i}_{j} : \mathcal{X} \to \mathbb{R}$ and $| \mathcal{I}^i | < \infty$. For agent $i$, we assume that the set of features are linearly independent. Let $K^{i}(x,y) = \sum_{j \in \mathcal{I}^i}\varphi^{i}_{j}(x)\varphi^{i}_{j}(y)$ be the kernel for agent $i$ and the corresponding RKHS generated by it be $(H^i, \langle \cdot, \cdot \rangle_{H^{i}}, ||\cdot ||_{ H^{i}})$. Then, the knowledge space at Agent $i$ is considered to be the RKHS, $H^i$, with kernel $K^{i}$, Figure \ref{Figure 1}. Given the downloaded function from iteration $n-1$, $\bar{f}^i_{n-1}$ at Agent $i$ and the data point $(x^{i}_n, y^{i}_n)$ at iteration $n$, each agent solves an estimation problem (subsection \ref{Subsection 2.2}) to arrive at the estimate $f^{i}_{n}$. Given the RKHS $H^1$ and $H^2$ with kernels $K^{1}$ and $K^{2}$, an RKHS with kernel $K^{1} + K^{2}$ can be constructed, Theorem \ref{Theorem 18}. The knowledge space at the fusion center, the fusion space is the RKHS, $H$, with kernel $K = K^{1} + K^{2}$. The locally estimated functions, $f^{1}_n$ and $f^{2}_n$ are uploaded to the fusion space using the operators, $\hat{L}^{1}$ and $\hat{L}^2$ respectively, Corollary \ref{Corollary 19}. The uploaded functions are fused in the fusion space (subsection \ref{Subsection 2.3}) to obtain $f_{n}$. The fused function is downloaded onto the KS of Agent $i$ using the download operator $\sqrt{\bar{L}^{i}} \circ \Pi_{\mathcal{N}\big(\sqrt{\bar{L}^i}\big)^{\perp}}$, Theorem \ref{Theorem 23}, and is denoted as $\bar{f}_{n}$. $\bar{f}_{n}$ is considered as the final estimate at the agents at iteration $n$. The Theorems and the associated corollaries invoked in the construction  of the knowledge spaces have been proved and discussed in detail in our previous work, \citep{raghavan2024distributed}, which is under review. Hence we only state the results in the Appendix (Section \ref{Section 6}).
\subsection{Estimation at the agents} \label{Subsection 2.2}
In this subsection, we discuss the estimation problem at the agents. At iteration $n$, given the data point $(x^{i}_{n},y^{i}_{n})$ and the downloaded function from iteration $n-1$, $\bar{f}^{i}_{n-1}$, the objective of each agent is to minimize the error square between the true output and the estimated output at the received input data point while simultaneously minimizing the norm square between the downloaded function at stage $n-1$ and the current estimate. Thus, the estimation problem for agent $i$ is, 
\begin{align*}
(P1)^{i}_{n}:\; \underset{f^{1}_{n} \in H^{i}} \min C^{i}(f^{i}_n), \;  C(f^{i}_n) = (y^{i}_{n} - f^{i}_{n}(x^{i}_{n}))^{2} +  \varrho^{i}_{n}|| f^{i}_{n} - \bar{f}^{i}_{n-1} ||^{2}_{H^i}
\end{align*}
\subsection{Fusion Problem}\label{Subsection 2.3}
From Corollary \ref{Corollary 19}, the function uploaded by Agent $i$ to $H$ is $f^{i}_{n}$, with different norm. Invoking Proposition \ref{Proposition 24}, the function uploaded by agent $i$ can be expressed as $f^{i}_{n} = \sum^{m}_{j=1}\alpha^{i}_{n,j}K^{i}(\cdot,\bar{x}^{i}_{j})$. Given $f^{1}_{n}$ and $f^{2}_{n}$, the first goal of the fusion center is to estimate the data points $\{(\hat{x}^{i}_{n,j},\hat{y}^{i}_{n,j})\}^{m}_{j=1}$ which under optimal estimation would result in the functions $f^{1}_{n}$ and $f^{2}_{n}$ being estimated. We formulate the problem as a feasibility problem:
\begin{align*}
(P2)^{i}_{n}: \;& \underset{\{(\hat{x}^{i}_{n,j},\hat{y}^{i}_{n,j})\}^{m}_{j=1}} \min \bar{c}^{i}\\
&\text{s.t} \;  f^{i}_n = \underset{g^{i}_{n} \in H^{i}} \argmin \sum^{m}_{j=1} (\hat{y}^{i}_{n,j} - g^{i}_{n}(\hat{x}^{i}_{n,j}))^2 + \varrho_{n}||g^{i}_{n}||^{2}
\end{align*}
Given $\{(\hat{x}^{1}_{n,j},\hat{y}^{1}_{n,j})\}^{m}_{j=1} \cup \{(\hat{x}^{2}_{n,j},\hat{y}^{2}_{n,j})\}^{m}_{j=1} $, the fusion problem as a least squares regression problem is:
\begin{align*}
(P3)_{n}:\; \underset{f_{n} \in H} \min \; C(f_{n}), \;  C(f_{n}) = \sum_{i=1,2} \sum^{m}_{j=1} (\hat{y}^{i}_{n,j} - f_{n}(\hat{x}^{i}_{n,j}))^2 + \varrho_{n} ||f_{n} ||^{2}_{H}.
\end{align*}
From Theorem \ref{Theorem 23}, it follows that the downloaded function at Agent $i$ is $\sqrt{\bar{L}^{i}} \circ \Pi_{\mathcal{N}\big(\sqrt{\bar{L}^i}\big)^{\perp}} (f_{n})$. From Proposition \ref{Proposition 24} it follows that, function downloaded can be expressed as $\bar{f}^{i}_{n} =\sum^m_{j=1}\bar{\alpha}^{i}_{n,j} K^{i}(\cdot, \bar{x}^{i}_{j})$ which is useful in solving problem $(P1)^{i}_n$ as demonstrated in Proposition \ref{Proposition 2}. 
\subsection{The Learning Algorithm}\label{Subsection 2.4}
The learning algorithm is described in Algorithm \ref{Algorithm 1}. Lines $1-10$ of the algorithm have been described while describing the learning architecture, in subsection \ref{Subsection 2.1}. Lines $11-21$, execute the stopping criterion. Essentially, it is verified that the difference in norm between the downloaded functions at any two iterations in $[n- k_{\max}, n]$ is upper bounded by $2\epsilon$. Consider $n_{1}, n_{2} \in [n- k_{\max}, n]$. 
\begin{align*}
|| \bar{f}^{1}_{n_1} -  \bar{f}^{1}_{n_2}||_{H^1} + || \bar{f}^{2}_{n_1} -  \bar{f}^{2}_{n_2}||_{H^2} \leq  \sum_{i=1,2}|| \bar{f}^{i}_{n_1} -\bar{f}^{i}_{n - k_{\max}} || + || \bar{f}^{i}_{n_2} - \bar{f}^{i}_{n - k_{\max}} || < 2\epsilon
\end{align*}
Thus, the stopping criterion approximately verifies that the sequence is Cauchy in norm.  
\begin{algorithm}
\caption{Distributed Learning Algorithm in RKHS}
\begin{algorithmic}[1]\label{Algorithm 1}
\STATE Initialize $\bar{f}^{1}_{0}, \bar{f}^{2}_{0}, \epsilon, \max, k_{\max}$ where $\max > \epsilon$
\STATE $n \gets 0$
\WHILE {$\max \geq \epsilon$}
\STATE $n \gets n+1$
\STATE Agent $i$ collects sample $(x^{i}_n; y^{i}_n)$. 
\STATE Agent $i$ solves optimization problem $(P1)^{i}_{n}$ to find $f^{i}_{n}$
\STATE Agent $i$ uploads $f^{i}_{n}$ to fusion space. 
\STATE Feasibility problem $(P2)^{i}_n$ is solved find to $\{(\bar{x}^{i}_{j}; \bar{y}^{i}_{j})\}^{m}_{j=1}$, $i=1,2$.
\STATE Fusion problem $(P3)_n$ is solved to find $f_n$.
\STATE $f_{n}$ is downloaded onto knowledge space of agent $i$, $H^{i}$ as $\bar{f}^{i}_n$.
\IF    {$n \geq k_{\max}$}
\STATE $j \gets 1$
\WHILE {$j \leq k_{\max}$}
\STATE $\text{temp} \; = \; || \bar{f}^{1}_{n - k_{\max} +j} -  \bar{f}^{1}_{n - k_{\max}}||+ || \bar{f}^{2}_{n - k_{\max} +j} -  \bar{f}^{2}_{n - k_{\max}}||$
\IF    {$ \text{temp} > \max$}
\STATE $\max \gets \text{temp}$
\ENDIF
\STATE $j \gets j+1$
\ENDWHILE
\ENDIF
\ENDWHILE
\end{algorithmic}
\end{algorithm}
\noindent We now define consistency of the learning scheme. 
\begin{definition}\label{Definition 1}
Given any pair of initial estimates $(f^{1};f^{2})$ and any sequence of data points $(\hspace{-1pt}\{(x^{1}_n; y^{1}_{n})\\ \}, \{(x^{2}_n; y^{2}_{n})\})$,  a learning scheme is said to be strongly consistent if there is a subsequence of the learned functions that converges strongly to a function in the knowledge space. 
\end{definition}
\begin{assumption}
We assume that $\mathcal{X}$ is a closed, connected set with no isolated points. We assume that the RKHS $H^{1}$ and $H^{2}$ are finite dimensional. 
\end{assumption}
In the following, whenever the results can be extended to infinite dimensional spaces, it has been highlighted. 
\section{Main Results}\label{Section 3}
In this section, we present solutions to the problems formulated in subsections \ref{Subsection 2.2} and \ref{Subsection 2.3}. We develop the theory, find sufficient conditions and prove results which are crucial to prove that Algorithm \ref{Algorithm 1} is consistent.
\subsection{Estimation at the Agents}\label{Subsection 3.1}
In this section, we present the solution to the estimation problem at the agents. We define and study the asymptotic properties of the learning operator at the agents. For each agent $i$, let $\mathbf{K^{i}} = (K^{i}(\bar{x}^{i}_{j}, \bar{x}^{i}_{k}))_{jk}= (\langle K^{i}(\cdot, \bar{x}^i_k), K^i(\cdot, \bar{x}^i_{j})  \rangle_{H^i})$. Since the inner product is symmetric, $\mathbf{K^{i}} \in \mathbb{R}^{m \times m}$ is symmetric matrix. Let $\mathbf{\bar{K}^{i}}: \mathcal{X} \to \mathbb{R}^m$ be defined as $\mathbf{\bar{K}^{i}}(\cdot) = [K^{i}(\cdot, \bar{x}^{i}_{1}), \ldots, K^{i}(\cdot, \bar{x}^{i}_{m})]$. For any data point, $\mathbf{\bar{K}^{i}}(x^{i}_n) = [K^{i}(x^{i}_{n}, \bar{x}^{i}_{1}); \ldots; K^{i}(x^{i}_{n}, \bar{x}^{i}_{m})]$ is a column vector in $\mathbb{R}^{m}$. Given vector $\boldsymbol{\alpha^i} = [\alpha^{i}_{1}; \ldots; \alpha^{i}_{m}] \in \mathbb{R}^{m}$, we use notation $f^{i} = \boldsymbol{\alpha^T} \mathbf{\bar{K}^{i}}(\cdot)$ for the function $f^{i} =\sum^{m}_{j=1} \alpha^{i}_{j}K(\cdot, \bar{x}^i_j) \in H^{i}.$
\begin{proposition}\label{Proposition 2}
Let $ \boldsymbol{\alpha^{1,*}_{n}} =  \Big( \varrho^{i}_{n}\mathbf{K^{i}}+ \mathbf{\bar{K}^{i}}(x^{i}_n) \mathbf{\bar{K}^{i^T}}(x^{i}_n) \Big)^{-1} \Big( \mathbf{\bar{K}^{i}}(x^{i}_n) y^{i}_{n} +  \varrho^{i}_{n}\mathbf{K^{i}} \boldsymbol{\bar{\alpha}^{1}_{n-1}}\Big)$. Then, $f^{i,*}_{n} = \boldsymbol{\alpha^{{1,*}^T}_{n}} \mathbf{\bar{K}^{i}}(\cdot)$ solves the optimization problem in problem $(P1)^{i}_n$. 
\end{proposition}
\begin{proof}
Using the notation, $f^{1}_n = \boldsymbol{\alpha^{1^T}_{n}}\mathbf{\bar{K}^{i}}(\cdot)$, the cost function of the estimation problem in subsection \ref{Subsection 2.1} can be expressed as $C^{i}(f^{i}_{n}) = (y^{1}_{n} - \boldsymbol{\alpha^{1}_{n}}^T\mathbf{\bar{K}^{i}}(x^{i}_n) )^2 + \varrho^{i}_{n}(\boldsymbol{\alpha^{1}_{n}} - \boldsymbol{\bar{\alpha}^{1}_{n-1}})^T\mathbf{K^{i}}(\boldsymbol{\alpha^{1}_{n}} - \boldsymbol{\bar{\alpha}^{1}_{n-1}})$. The gradient of $C^{i}(f^{i}_{n})$ with respect to $\boldsymbol{\alpha^{1}_{n}}$ is,
\begin{align*}
\nabla_{\boldsymbol{\alpha^{1}_{n}}} C^{i}(f^{i}_{n}) =  - 2(y^{1}_{n} - \boldsymbol{\alpha^{1}_{n}}^T\mathbf{\bar{K}^{i}}(x^{i}_n) )\mathbf{\bar{K}^{i}}(x^{i}_n) +  \varrho^{i}_{n} \mathbf{K^{i}}^T(\boldsymbol{\alpha^{1}_{n}} - \boldsymbol{\bar{\alpha}^{1}_{n-1}}) +  \varrho^{i}_{n} \mathbf{K^{i}}(\boldsymbol{\alpha^{1}_{n}} - \boldsymbol{\bar{\alpha}^{1}_{n-1}}). 
\end{align*}
It can be verified that, $\Big(\boldsymbol{\alpha^{1}_{n}}^T\mathbf{\bar{K}^{i}}(x^{i}_n) \Big)\mathbf{\bar{K}^{i}}(x^{i}_n) = \mathbf{\bar{K}^{i}}(x^{i}_n) \mathbf{\bar{K}^{i^T}}(x^{i}_n) \boldsymbol{\alpha^{1}_{n}}$. Thus,
\begin{align*}
\nabla_{\boldsymbol{\alpha^{1}_{n}}} C^{i}(f^{i}_{n}) = 2( \varrho^{i}_{n}\mathbf{K^{i}}+ \mathbf{\bar{K}^{i}}(x^{i}_n) \mathbf{\bar{K}^{i^T}}(x^{i}_n) )  \boldsymbol{\alpha^{1}_{n}} - 2(\mathbf{\bar{K}^{i}}(x^{i}_n) y^{i}_{n} +  \varrho^{i}_{n}\mathbf{K^{i}} \boldsymbol{\bar{\alpha}^{1}_{n-1}}). 
\end{align*}
By setting the gradient to zero, we obtain the desired result.
\end{proof} 
Given data point $(x,y)$ and current knowledge as $f$, we define the learning operator at agent $i$, $T^{i}: H^{i} \to H^{i} $ as $T^{i}(x,y, \varrho)[f] = \underset{g \in H^{i}}\argmin \; (y -g(x))^2 + \varrho|| f - g ||^2_{H^{i}}$.  Let $M: \mathcal{X} \times \mathbb{R} \to \mathbb{R}^{m \times m}$ be defined as $M^i(x, \varrho) =  \Big( \varrho\mathbf{K^{i}}+ \mathbf{\bar{K}^{i}}(x) \mathbf{\bar{K}^{i^T}}(x) \Big)^{-1}$. Then, $T^{i}(x,y,\varrho)[f] =\Big(M^{i}(x, \varrho)\big(\varrho K^{i}\boldsymbol{\alpha} + (\mathbf{\bar{K}^{i}}(x)y)^T\big)\Big)^{T}\mathbf{\bar{K}^{i}} (\cdot)$ is an affine mapping in $\boldsymbol{\alpha^T}\mathbf{\bar{K}^{i}}(\cdot) = f$. In order to avoid the complications associated with the definition of a norm for a nonlinear operator and its properties specifically with respect to composition of operators, we define the learning operator as follows:
\begin{definition}\label{Definition 3}
Given data point, $(x,y)$, we define $\psi^{i}_{(x,y)}$ as $ \psi^{i}_{(x,y)} = (\mathbf{\bar{K}^{i}}(x)y)^T\mathbf{\bar{K}^{i}}(\cdot)$. The learning operator at agent $i$, $\bar{T}^{i}: H^{i} \times H^{i} \to H^{i}$ is defined as $\bar{T}^{i}(\varrho)[ f, \psi^{i}_{(x,y)} ] =   \Big(M^{i}(x, \varrho)\big(\varrho K^{i}\boldsymbol{\alpha} + \mathbf{\bar{K}^{i}}(x)y\big)\Big)^{T}\mathbf{\bar{K}^{i}} (\cdot)$, where $f= \boldsymbol{\alpha^T}\mathbf{\bar{K}^{i}} (\cdot)$. $\bar{T}^{i}(\varrho)[\cdot]$ is linear and bounded.
\end{definition}
\begin{proposition}\label{Proposition 4}
Let $\{\varrho^i_{n}\}$ be sequence of real numbers such that $\varrho^i_{n} \to \infty$. Then, there is a subsequence $\{\varrho^i_{n_k}\}$ of $\{\varrho^i_{n}\}$ such that $\underset{k \to \infty} \lim || \bar{T}^{i}(\varrho^i_{n_k})|| = 1$.
\end{proposition}
\begin{proof}
Let $f =  \boldsymbol{\alpha^{T}}\mathbf{\bar{K}^{i}} (\cdot)$ and $\psi^{i}_{(x,y)}$ be as defined above. The set $\{f \in H^{i},  x \in \mathcal{X}, y \in \mathbb{R}\;:\: || f ||^{2} + || \psi^{i}_{(x,y)} ||^{2} =1 \}$ is isomorphic to the set $E^{i} =\{\boldsymbol{\alpha} \in \mathbb{R}^{m}, x \in \mathcal{X}, y \in \mathbb{R}: \boldsymbol{\alpha^{T}}\mathbf{K^{i}}\boldsymbol{\alpha^{T}} + y^2 \mathbf{\bar{K}^{i^T}}(x)\mathbf{K^{i}} \mathbf{\bar{K}^{i}}(x)= 1 \}$. $E^{i}$ is a compact subset of $\mathbb{R}^m \times \mathbb{R}^I \times \mathbb{R}$. For any $n$, 
\begin{align*}
&|| \bar{T}^{i}(\varrho^i_{n})|| = \underset{(f, \psi^{i}_{(x,y)}) \neq \theta} \sup \frac{\Big(M^{i}(x, \varrho^i_n)\big(\varrho^i_n K^{i}\boldsymbol{\alpha} + \mathbf{\bar{K}^{i}}(x)y\big)\Big)^{T} K^{i} \Big(M^{i}(x, \varrho^i_n)\big(\varrho^i_n K^{i}\boldsymbol{\alpha} + \mathbf{\bar{K}^{i}}(x)y\big)\Big)}{|| f ||^{2} + || \psi^{i}_{(x,y)} ||^{2}}\\
&\phi^{i}_{n}(\boldsymbol{\alpha};x;y) = \underbrace{\Big(\frac{\mathbf{\bar{K}^{i}}(x) y}{\varrho^{i}_{n}} + \mathbf{K^{i}} \boldsymbol{\alpha}\Big)^{T}}_{\phi^{i^T}_{n,1}} \underbrace{\Big(\mathbf{K^{i}} +\frac{\mathbf{\bar{K}^{i}}(x) \mathbf{\bar{K}^{i^T}}(x) }{\varrho^{i}_{n}}\Big)^{-1}}_{\phi^{i}_{n,2}} \mathbf{K^{i}} \Big(\mathbf{K^{i}}+ \frac{\mathbf{\bar{K}^{i}}(x) \mathbf{\bar{K}^{i^T}}(x) }{\varrho^{i}_{n}}\Big)^{-1} \\
&\times \Big( \frac{\mathbf{\bar{K}^{i}}(x) y}{\varrho^{i}_{n}} +  \mathbf{K^{i}} \boldsymbol{\alpha}\Big).\; || \bar{T}^{i}(\varrho^i_{n})|| = \underset{(f, \psi^{i}_{(x,y)}) \neq \theta} \sup  \frac{\phi^{i}_{n}(\boldsymbol{\alpha},x,y)}{|| f ||^{2} + || \psi^{i}_{(x,y)} ||^{2}} = \underset{(\boldsymbol{\alpha} , x , y) \in E^i} \sup  \phi^{i}_{n}(\boldsymbol{\alpha};x;y).
\end{align*}
First, we claim that, $\{\phi^{i}_{n,1}(\cdot)\}$ is an equicontinuous sequence of functions on $E^{i}$. Let $(\boldsymbol{\alpha_0};x_0;y_0) \in E^{i}$. Since $E^{i}$ is compact, $\exists \; c_{1} < \infty$ such that $|y| < c_{1}, y \in E^{i}$. Since the sequence $\{\frac{1}{\varrho_{n}}\}$ is bounded, $\exists \; c_{2} < \infty$ such that $|\frac{1}{\varrho_{n}}| \leq c_{2} \forall n$. Since $K^{i}(\cdot,\bar{x}^i_j)$ is uniformly continuous on $E^{i}$ for each $j$, $\mathbf{\bar{K}^{i}}(\cdot)$ is uniformly continuous on $E^{i}$. Thus, $\forall \epsilon > 0$, $\exists \delta_{1} >0$ such that $|| x- x_{0}|| < \delta_1$ implies $|| \mathbf{\bar{K}^{i}}(x) -  \mathbf{\bar{K}^{i}}(x_0) || < \frac{\epsilon}{2 \times c_{1} \times c_{2}} \implies || \frac{\mathbf{\bar{K}^{i}}(x) y}{\varrho^{i}_{n}}- \frac{\mathbf{\bar{K}^{i}}(x_0) y}{\varrho^{i}_{n}}|| <  \frac{\epsilon}{2}$. Since $g(\alpha) = \mathbf{K^{i}} \boldsymbol{\alpha}$ is linear in $\boldsymbol{\alpha}$, $\forall \epsilon > 0$, let $0 < \delta_2 < \frac{\epsilon}{2 \times \lambda_{\max}( \mathbf{K^{i}})}$, then, $|| \boldsymbol{\alpha} - \boldsymbol{\alpha_0} || < \delta_2$ implies $|| \mathbf{K^{i}} \boldsymbol{\alpha} - \mathbf{K^{i}} \boldsymbol{\alpha_0} || = ||  \mathbf{K^{i}}(\boldsymbol{\alpha} -  \boldsymbol{\alpha_0} ) || < \frac{\epsilon}{2}$. Let $\delta = \min(\delta_1, \delta_2)$. Consider any vector in $(\boldsymbol{\alpha_0};x_0;y_0) \in E^{i}$ and let $(\boldsymbol{\alpha};x;y) \in E^{i}$ be such that $|| (\boldsymbol{\alpha};x;y) - (\boldsymbol{\alpha_0};x_0;y_0) || < \delta$. Then $|| \phi^{i}_{n,1}(\boldsymbol{\alpha};x;y) -  \phi^{i}_{n,1}(\boldsymbol{\alpha_0};x_0;y_0) || <  || \frac{\mathbf{\bar{K}^{i}}(x) y}{\varrho^{i}_{n}}- \frac{\mathbf{\bar{K}^{i}}(x_0) y}{\varrho^{i}_{n}}|| + || \mathbf{K^{i}} \boldsymbol{\alpha} - \mathbf{K^{i}} \boldsymbol{\alpha_0} || < \epsilon \; \forall n$. Thus,  $\{\phi^{i}_{n,1}(\cdot)\}$ is uniformly equicontinuous on $E^{i}$. Further, it is clear that $\{\phi^{i}_{n,1}(\cdot)\}$ is uniformly bounded on $E_{i}$, i.e., $\exists c_{\phi^{i}_{n,1}}$ such that $|| \phi^{i}_{n,1}(\boldsymbol{\alpha};x;y) || <  c_{\phi^{i}_{n,1}} \forall  (\boldsymbol{\alpha};x;y) \in E^{i}, \forall n$.\\
Second, we claim that, $\{\phi^{i}_{n,2}(\cdot)\}$ is an equicontinuous sequence of functions on $E^{i}$. We recall that the map from $GL_{m}(\mathbb{R}) \to M_{m}(\mathbb{R})$ given by  $M \mapsto M^{-1}$ is continuous in the operator norm topology. Let $M \in GL_{n}(\mathbb{R})$. Define, $S= M^{-1} \sum^{\infty}_{n=0} (M -N)^nM^{-n}$, $N \in M_{m}(\mathbb{R})$. The sequence in this definition converges if $||M -N || \times ||M^{-1}|| < 1$ as the space $M_{m}(\mathbb{R})$ is complete. Let $N$  be such that $||M -N || < \frac{1}{||M^{-1}||}$. Then, $S - S(M-N)M^{-1} = M^{-1} \implies S(\mathbb{I}_{m} - (M-N)M^{-1}) = M^{-1} \implies S= M^{-1}(\mathbb{I}_{m} - (M-N)M^{-1})^{-1} = ((\mathbb{I}_{m} - (M-N)M^{-1})M)^{-1} = N^{-1}$. $\mathbb{I}_{m} - (M-N)M^{-1}$ is invertible as $|| \mathbb{I}_{m} - (M-N)M^{-1} || \geq  ||\mathbb{I}_{m}|| - ||(M-N)M^{-1}|| > 0 $. The expansion of $N^{-1}$ gives the bound $|| N^{-1} || \leq \frac{||M^{-1}||}{1 - || M - N|| \; ||M^{-1}|| }$. Thus, $N \in B_{\frac{1}{|| 2 \times M^{-1}||}}(M)$ is invertible and $|| N^{-1} || \leq 2 ||M^{-1}||$. For such N, $ || N^{-1} - M^{-1} || = ||M^{-1} (M -N) N^{-1} ||\leq 2 ||M^{-1} ||^{2} ||M - N ||$. Given $\epsilon$, let $\delta < \min (\frac{1}{2 \times || M^{-1}||}, \frac{\epsilon}{2||M^{-1} ||^{2}})$. Then, $|| N - M || < \delta \implies || N^{-1} - M^{-1} || < \epsilon$. Since the map $M \mapsto M^{-1}$ is nonlinear map, by the $\epsilon -\delta$ definition of continuity, the map is continuous. We note that $\mathbf{K^{i}} +\frac{\mathbf{\bar{K}^{i}}(x) \mathbf{\bar{K}^{i^T}}(x) }{\varrho^{i}_{n}}$ need \textit{not} be invertible $\forall n$ since $\mathbf{K^{i}}$ is only positive semidefinite. In case $\mathbf{K^{i}}$ is not invertible, it is replaced with $\mathbf{K^{i}} + c^{i}_{3}\mathbb{I}_{m}$, where $c^{i}_{3}$ is small positive constant, to ensure that the previous sequence of matrices is invertible. Since $\mathbf{\bar{K}^{i}}(\cdot)$ is uniformly continuous on $E^{i}$, given $\epsilon >0$, $\exists \delta >0$ such that $|| x- x_{0} || < \delta$ implies $||\mathbf{\bar{K}^{i}}(x) \mathbf{\bar{K}^{i^T}}(x)  - \mathbf{\bar{K}^{i}}(x_0) \mathbf{\bar{K}^{i^T}}(x_0) || < \frac{\epsilon}{c_{2}}$. Thus, $|| x- x_{0} || < \delta$ implies $|| \mathbf{K^{i}} +\frac{\mathbf{\bar{K}^{i}}(x) \mathbf{\bar{K}^{i^T}}(x) }{\varrho^{i}_{n}} - \mathbf{K^{i}} - \frac{\mathbf{\bar{K}^{i}}(x_0) \mathbf{\bar{K}^{i^T}}(x_0) }{\varrho^{i}_{n}} || = ||\frac{\mathbf{\bar{K}^{i}}(x) \mathbf{\bar{K}^{i^T}}(x)  - \mathbf{\bar{K}^{i}}(x_0) \mathbf{\bar{K}^{i^T}}(x_0)}{\varrho^{i}_{n}}|| < \epsilon, \forall n$. For $x \in E^{i}$, let $\underset{n \in \mathbb{N}} \sup \; || \mathbf{K^{i}} + \frac{\mathbf{\bar{K}^{i}}(x_0) \mathbf{\bar{K}^{i^T}}(x_0)}{\varrho^{i}_{n}} || = c_{4}(x_0) > 0$. Given $\epsilon >0$, let $\delta >0 $ be such that  $|| x- x_{0} || < \delta$ implies $||\frac{\mathbf{\bar{K}^{i}}(x) \mathbf{\bar{K}^{i^T}}(x)  - \mathbf{\bar{K}^{i}}(x_0) \mathbf{\bar{K}^{i^T}}(x_0)}{\varrho^{i}_{n}}|| < \frac{\epsilon}{2 c^2_{4}(x_0)}, \forall n$. For $x$ such that $|| x- x_{0} || < \delta$, $|| \phi^{i}_{n,2}(\boldsymbol{\alpha};x;y) - \phi^{i}_{n,2}(\boldsymbol{\alpha}_0;x_0;y_0)|| \leq 2 c^2_{4}(x_0) \times ||\mathbf{K^{i}} +\frac{\mathbf{\bar{K}^{i}}(x) \mathbf{\bar{K}^{i^T}}(x) }{\varrho^{i}_{n}} - \mathbf{K^{i}} - \frac{\mathbf{\bar{K}^{i}}(x_0) \mathbf{\bar{K}^{i^T}}(x_0) }{\varrho^{i}_{n}}  || <  \epsilon, \forall n$. Further, we note that $\{\phi^{i}_{n,2}(\cdot)\}$ is uniformly bounded on $E_{i}$, i.e., $\exists c_{\phi^{i}_{n,2}}$ such that $|| \phi^{i}_{n,2}(\boldsymbol{\alpha};x;y) || <  c_{\phi^{i}_{n,2}} \forall  (\boldsymbol{\alpha};x;y) \in E^{i}, \forall n$. \\
Third, we claim that sequence $\{\phi^{i}_{n,2}(\cdot)\phi^{i}_{n,1}(\cdot)\}$ is equicontinuous. Indeed, let $(\boldsymbol{\alpha_0};x_0;y_0) \in E^{i}$. Given $\epsilon >0$, there exists $\delta>0$ such that $||(\boldsymbol{\alpha};x;y) - (\boldsymbol{\alpha_0};x_0;y_0) || < \delta \implies || \phi^{i}_{n,1}(\boldsymbol{\alpha};x;y) - \phi^{i}_{n,1}(\boldsymbol{\alpha};x;y) || < \frac{\epsilon}{2 \times c_{\phi^{i}_{n,2}}}$ and  $|| \phi^{i}_{n,2}(\boldsymbol{\alpha};x;y) - \phi^{i}_{n,2}(\boldsymbol{\alpha};x;y) || < \frac{\epsilon}{2 \times c_{\phi^{i}_{n,1}}}, \forall n$. Then, $||(\boldsymbol{\alpha};x;y) - (\boldsymbol{\alpha_0};x_0;y_0) || < \delta$, implies
\begin{align*}
&|| \phi^{i}_{n,2}(\boldsymbol{\alpha};x;y)\phi^{i}_{n,1}(\boldsymbol{\alpha};x;y) - \phi^{i}_{n,2}(\boldsymbol{\alpha_0};x_0;y_0)\phi^{i}_{n,1}(\boldsymbol{\alpha_0};x_0;y_0)|| = || \phi^{i}_{n,2}(\boldsymbol{\alpha};x;y)\phi^{i}_{n,1}(\boldsymbol{\alpha};x;y) - \\
&\phi^{i}_{n,2}(\boldsymbol{\alpha_0};x_0;y_0)\phi^{i}_{n,1}(\boldsymbol{\alpha};x;y) + \hspace{-1pt}\phi^{i}_{n,2}(\boldsymbol{\alpha_0};x_0;y_0) \phi^{i}_{n,1}(\boldsymbol{\alpha};x;y) - \hspace{-1pt} \phi^{i}_{n,2}(\boldsymbol{\alpha_0};x_0;y_0)\phi^{i}_{n,1}(\boldsymbol{\alpha_0};x_0;y_0)||\\
&\leq || \phi^{i}_{n,2}(\boldsymbol{\alpha};x;y) ||\; || \phi^{i}_{n,1}(\boldsymbol{\alpha};x;y) - \phi^{i}_{n,1}(\boldsymbol{\alpha_0};x_0;y_0) || + ||  \phi^{i}_{n,1}(\boldsymbol{\alpha_0};x_0;y_0)||\; || \phi^{i}_{n,2}(\boldsymbol{\alpha};x;y)  -\\
&\phi^{i}_{n,2}(\boldsymbol{\alpha_0};x_0;y_0)|| <  c_{\phi^{i}_{n,2}} \times  \frac{\epsilon}{2 \times c_{\phi^{i}_{n,2}}} + c_{\phi^{i}_{n,1}} \times \frac{\epsilon}{2 \times c_{\phi^{i}_{n,1}}} < \epsilon, \; \forall n. 
\end{align*} 
$|| \phi^{i}_{n,2}(\boldsymbol{\alpha};x;y)\phi^{i}_{n,1}(\boldsymbol{\alpha};x;y) || < c_{\phi^{i}_{n,2}} c_{\phi^{i}_{n,1}}\forall  (\boldsymbol{\alpha};x;y) \in E^{i}, \forall n$. Thus, the product of uniformly bounded equicontinuous sequence of functions is a uniformly bounded equicontinuous sequence of functions. By this argument the sequences,  $\{\mathbf{K^{i}}\phi^{i}_{n,2}(\cdot)\phi^{i}_{n,1}(\cdot)\}$, $\{\phi^{i^T}_{n,2}(\cdot)\mathbf{K^{i}}\phi^{i}_{n,2}(\cdot)\phi^{i}_{n,1}(\cdot)\}$, $\{ \phi^{i^T}_{n,1}(\cdot)\\ \phi^{i^T}_{n,2}(\cdot) \mathbf{K^{i}}\phi^{i}_{n,2}(\cdot)\phi^{i}_{n,1}(\cdot)\}$ are equicontinuous on $E^{i}$. Thus, $\{\phi^{i}_{n}(\cdot)\}$ is sequence of equicontinuous functions on $E^{i}$ which converges pointwise to $\boldsymbol{\alpha}^{T}\mathbf{K^{i}}\boldsymbol{\alpha}\; \text{on}\; E^{i}$. By the \textit{Arzel\`a Ascoli Theorem}, there exists a subsequence $\{\phi^{i}_{n_k}(\cdot)\}$ that converges uniformly to $\boldsymbol{\alpha}^{T}\mathbf{K^{i}}\boldsymbol{\alpha}\; \text{on}\; E^{i}$. Hence the limit and the supremum can be swapped as below, 
\begin{align*}
\underset{k \to \infty}\lim || \bar{T}^{i}(\varrho^{i}_{n_k})||^2 = \underset{k \to \infty}\lim \; \underset{ E^{i}} \sup \; \phi^{i}_{n_k}(\boldsymbol{\alpha};x;y) = \underset{E^{i}} \sup \; \underset{k \to \infty}\lim \; \phi^{i}_{n_k}(\boldsymbol{\alpha};x;y) = \underset{E^{i}} \sup \; \boldsymbol{\alpha}^{T}\mathbf{K^{i}}\boldsymbol{\alpha} = 1.
\end{align*}
\end{proof}
\subsection{The Fusion Problem}\label{Subsection 3.2}
Let $\mathbf{K} =  (K(\bar{x}^{p}_{j}, \bar{x}^{q}_{k}))_{jk}= (\langle K(\cdot, \bar{x}^q_k), K(\cdot, \bar{x}^p_{j}) \rangle_{H})_{j,k}, p,q=1,2, j,k, =1 \ldots m$ be a symmetric matrix in $\mathbb{R}^{2m \times 2m}$ and $\mathbf{\hat{K}^{i}_n} = (K^{i}(\hat{x}^{i}_{n,j}, \hat{x}^{i}_{n,k}))_{jk}= (\langle K^{i}(\cdot, \hat{x}^i_{n,k}), K^i(\cdot, \hat{x}^i_{n,j})  \rangle_{H^i})_{j,k} \in \mathbb{R}^{m \times m}$ be a symmetric matrix. Let $\mathbf{\bar{K}}: \mathcal{X} \to \mathbb{R}^{2m}$ be defined as $\mathbf{\bar{K}}(\cdot) = \hspace{-2pt} \Big[K(\cdot, \bar{x}^{1}_{1}); \ldots,  K(\cdot, \bar{x}^{1}_{m});  K(\cdot, \bar{x}^{2}_{1});\ldots ;  \\ K(\cdot, \bar{x}^{2}_{m})\Big]$.  Let $\mathbf{\tilde{K}^{i}_{n}}: \mathcal{X} \to \mathbb{R}^m$ be defined as $\mathbf{\tilde{K}^{i}_{n}}(\cdot) = [K(\cdot, \hat{x}^{i}_{n,1}), \ldots, K(\cdot, \hat{x}^{i}_{n,m})]$. Let $\mathbf{\hat{Y}}^i_n = \Big[\hat{y}^{i}_{n,1}; \ldots; \hat{y}^{i}_{n,m}\Big] \in \mathbb{R}^m$ and $\mathbf{\hat{Y}}_{n} =  \Big[\mathbf{\hat{Y}}^1_{n} ;  \mathbf{\hat{Y}}^2_{n}\Big] \in \mathbb{R}^{2m}$. Let $\mathbf{\check{K}^{i}}: \mathcal{X} \to \mathbb{R}^{m}$ be defined as $\mathbf{\check{K}^{i}}(\cdot) = \Big[K(\cdot, \bar{x}^{i}_{1}); \ldots,  K(\cdot, \bar{x}^{i}_{m})\Big]$ and $\mathbf{\breve{K}^{i}}(\cdot) = (K(\bar{x}^{i}_{j}, \bar{x}^{i}_{k}))_{jk}= (\langle K(\cdot, \bar{x}^i_k), K(\cdot, \bar{x}^i_{j}) \rangle_{H})_{j,k}, i=1,2,\; j,k = 1 \ldots m$ be a symmetric matrix in $\mathbb{R}^{m \times m}$. A function $f^{i}$ which belongs to $H^{i}$ and $H$ can be expressed as $\boldsymbol{\alpha^{i^T}} \mathbf{\bar{K}^{i}}(\cdot)$ and $\boldsymbol{\hat{\alpha}^{i^T}}\mathbf{\check{K}^{i}}(\cdot)$. Since $\{K(\cdot, \bar{x}^i_j)\}^{m}_{j=1}$ is basis for $H$, $\exists! \{M^{i}_{kj}\}^{m}_{k=1}$ such that $K^{i}(\cdot, \bar{x}^{i}_j) =\sum^{m}_{k=1}M_{kj}K(\cdot, \bar{x}^i_k), \forall j$. Thus, $\sum^{m}_{j=1}\alpha^{i}_{j}K^{i}(\cdot, \bar{x}^{i}_j) = \sum^{m}_{k=1}\sum^{m}_{j=1}M_{kj}\alpha^{i}_{j}K(\cdot, \bar{x}^i_k) = \sum^{m}_{k=1}\hat{\alpha}_{k}K(\cdot, \bar{x}^i_k)$, which implies that $\boldsymbol{\hat{\alpha}^{i}} = M^{i}\boldsymbol{\alpha^{i}}$ where $M^{i} = (M_{kj})_{k,j} \in \mathbb{R}^{m \times m}$. In the following, we solve the problems $(P2)^i_n$ and $(P3)_n$. We define and the study the asymptotics of the learning operator at the fusion center.  
\begin{proposition}\label{Proposition 5}
$\hat{x}^{i}_{n,j} = \bar{x}^{i}_{j}\; \forall n \in \mathbb{N}, i=1,2, j= 1, \ldots, m$ and $\mathbf{\hat{Y}}^{i}_{n} = \Big(\mathbf{K^{i}} + \varrho^{i}_{n}\mathbb{I}_{m}\Big) \boldsymbol{\alpha^{i}_{n}}$ solve the problem $(P2)^{i}_{n}$ where $f^{i}_{n} = \boldsymbol{\alpha^{i^T}_{n}}\mathbf{\bar{K}^{i}}(\cdot)$. $f_{n} = \boldsymbol{\alpha^{T}_{n}}\mathbf{\bar{K}}(\cdot)$, where $\boldsymbol{\alpha_{n}} = \Big(\mathbf{K^{T}}\mathbf{K} + \varrho_{n}\mathbf{K}\Big)^{-1}\Big(\mathbf{K^{T}}\mathbf{\hat{Y}}_{n}\Big)$ solves $(P3)_{n}$.
\end{proposition}
\begin{proof}
If Agent $i$ had received the data points $\{(\hat{x}^{i}_{n,j},\hat{y}^{i}_{n,j})\}^{m}_{j=1}$, by the \textit{Representer Theorem}, \cite{hofmann2008kernel}, the optimal solution for a least squares regression problem for Agent $i$ is given by $g^{i}_{n} = \boldsymbol{\bar{\alpha}^{i^T}_{n}}\mathbf{\tilde{K}^{i}_{n}}(\cdot)$, where $\boldsymbol{\bar{\alpha}^{i}_{n}} = \Big(\mathbf{\hat{K}^{i^T}_{n}}\mathbf{\hat{K}^{i}_{n}} + \varrho_{n}\mathbf{\hat{K}^{i}_{n}}\Big)^{-1}\Big(\mathbf{\hat{K}^{i^T}_n}\mathbf{\hat{Y}^i_n}\Big)$. From Proposition \ref{Proposition 24}, the uploaded function is of the form $f^{i}_{n} = \boldsymbol{\alpha^{i}_{n}}\mathbf{\bar{K}^{i}}(\cdot)$. The feasibility problem in $(P2)^{i}_{n}$ is solved if and only if $f^{i}_{n}  = g^{i}_{n}$. To achieve the same we let, $\mathbf{\bar{K}^{i}}(\cdot) = \mathbf{\tilde{K}^{i}_{n}}(\cdot), \forall n$ and $\boldsymbol{\bar{\alpha}^{i}_{n}} = \boldsymbol{\alpha^{i}_{n}}, \forall n$. If $\hat{x}^{i}_{n,j} = \bar{x}^{i}_{j}$, then $\mathbf{\bar{K}^{i}}(\cdot) = \mathbf{\tilde{K}^{i}_{n}}(\cdot)$. This implies that $\mathbf{\hat{K}^{i}_n} = \mathbf{K^{i}}$. Using $\boldsymbol{\bar{\alpha}^{i}_{n}} = \boldsymbol{\alpha^{i}_{n}}$, $\mathbf{\hat{Y}}^i_{n}$ is obtained as $\mathbf{\hat{Y}}^i_{n} = \mathbf{K^{i^{T^{-1}}}}\Big(\mathbf{K^{i^T}}\mathbf{K^{i}} + \varrho_{n}\mathbf{K^{i}} \Big) \boldsymbol{\alpha^{i}_{n}} = \Big(\mathbf{K^{i}} + \varrho_{n}\mathbb{I}_{m}\Big) \boldsymbol{\alpha^{i}_{n}}$. Given the gram matrix generated from the kernel $K$ at the fusion input data points, $\{\bar{x}^{1}_{j}\}^{m}_{j=1} \cup \{\bar{x}^{2}_{j}\}^{m}_{j=1}$, $\mathbf{K}$,  and the fusion output data points $\mathbf{\hat{Y}}_{n}$, the solution of the estimation problem, $(P3)_{n}$,  is given by the \textit{Representer Theorem}, $f^{*}_{n} = \boldsymbol{\alpha^T_{n}}\mathbf{\bar{K}}$, where $\boldsymbol{\alpha_{n}} = \Big(\mathbf{K^{T}}\mathbf{K} + \varrho_{n}\mathbf{K}\Big)^{-1}\Big(\mathbf{K^{T}}\mathbf{\hat{Y}}_{n}\Big)$. 
\end{proof}
In the title of this paper, we mention ``A forget and relearn approach"; ``forget" meaning that the data gathered by the agents locally is forgotten or hidden in the functions transmitted by them. ``Relearn" meaning; at the fusion center, the data is retrieved from the received functions and the function is estimated from the retrieved data. 
\begin{definition}
Given the functions $f^{1} \in \{H \cap H^{1}\}$ and $f^{2} \in \{H\cap H^{2}\}$ uploaded to fusion center, we define the learning operator at the fusion center $T : \{H\cap H^{1}\} \times \{H \cap H^{2}\}  \to H$ as follows:
\begin{align*}\label{Definition 6}
T(\varrho)[f^{1},f^{2}] = \boldsymbol{\alpha^T}\mathbf{\bar{K}}, \boldsymbol{\alpha} = \Big(\mathbf{K^{T}}\mathbf{K} + \varrho\mathbf{K}\Big)^{-1}\Bigg[&\mathbf{K^{T}} 
\begin{bmatrix}
\Big(\mathbf{K^{1}} + \varrho \mathbb{I}_{m}\Big) \boldsymbol{\alpha^{1}} \\
\Big(\mathbf{K^{2}} + \varrho\mathbb{I}_{m} \Big) \boldsymbol{\alpha^{2}}
\end{bmatrix} \Bigg], 
\end{align*}
where $f^{i} = \boldsymbol{\alpha^{i}}\mathbf{\bar{K}^{i}}(\cdot) = \boldsymbol{\hat{\alpha}^{i^T}}\mathbf{\check{K}^{i}}(\cdot)= \boldsymbol{\alpha^{i}}M^{i^T}\mathbf{\check{K}^{i}}(\cdot)$. 
\end{definition}
We note that the definition of $T$ is restricted to subspaces of $H$ which contain and are contained in $H^{1}$ or $H^{2}$. This translates into the definition of the norm as well as in the proof below.
\begin{proposition}\label{Proposition 7}
Let $\mathbf{K} =\begin{bmatrix} 
&\mathbf{\tilde{K}^{1}} &\mathbf{K^{12}} \\ 
&\mathbf{K^{{12}^T}} &\mathbf{\tilde{K}^{2}} 
\end{bmatrix}$ be such that $\max(\lambda_{\max}(\mathbf{\tilde{K}^{1}}), \lambda_{\max}(\mathbf{\tilde{K}^{2}})) \leq 1$ and $\varrho_{n}$ be monotone sequence of positive real numbers diverging to $\infty$. Then, $\underset{ n \to \infty} \lim || T(\varrho_n) || \leq 1$. 
\end{proposition}
\begin{proof}
The set $\{(f^1;f^2) \in \{H \cap H^{1}\} \times \{ H \cap H^{2}\} : || f^{1} ||^2_{H} + || f^{2} ||^2_{H} =1 \}$ is isomorphic to the set $ E = \{\boldsymbol{\hat{\alpha}} = (\boldsymbol{\hat{\alpha}^{1}},\boldsymbol{\hat{\alpha}^{2}}) \in \mathbb{R}^{2m}: \exists \boldsymbol{\alpha^{i}} \text{ s.t } \boldsymbol{\hat{\alpha}^{i}} = M^{i}\boldsymbol{\alpha^{i}}, \text{ and }    \boldsymbol{\hat{\alpha}^{1^T}}\mathbf{\breve{K}^{1}} \boldsymbol{\hat{\alpha}^{1}} + \boldsymbol{\hat{\alpha}^{2^T}}\mathbf{\breve{K}^{2}} \boldsymbol{\hat{\alpha}^{2}} = 1 \} = \{ \boldsymbol{\alpha} = (\boldsymbol{\alpha^{1}},\boldsymbol{\alpha^{2}}) \in \mathbb{R}^{2m}: \boldsymbol{\alpha^{1^T}}M^{1^T}\mathbf{K^{1}} M^{1}\boldsymbol{\alpha^{1}} + \boldsymbol{\alpha^{2^T}}M^{2^T}\mathbf{K^{2}} M^{2}\boldsymbol{\alpha^{2}}  = 1\}$. Define $\phi_{n}(\cdot)$ as follows,
\begin{align*}
\phi_{n}(\boldsymbol{\alpha^{1}};\boldsymbol{\alpha^{2}}) = \bigg[\Big[\boldsymbol{\alpha^{1^T}} \Big(\frac{\mathbf{K^{1}}}{\varrho_n} +  \mathbb{I}_{m}\Big), \boldsymbol{\alpha^{2^T}} \Big(\frac{\mathbf{K^{2}}}{\varrho_n} +  &\mathbb{I}_{m}\Big) \Big] \mathbf{K} \bigg] \Big(\frac{\mathbf{K^{T}}\mathbf{K}}{\varrho_n} + \mathbf{K}\Big)^{-1} \mathbf{K}\\
&\underbrace{\Big(\frac{\mathbf{K^{T}}\mathbf{K}}{\varrho_n} + \mathbf{K}\Big)^{-1}}_{\phi_{n,2}(\cdot)}\underbrace{\Bigg[\mathbf{K^{T}} 
\begin{bmatrix}
\Big(\frac{\mathbf{K^{1}}}{\varrho_n} +  \mathbb{I}_{m}\Big) \boldsymbol{\alpha^{1}} \\
\Big(\frac{\mathbf{K^{2}}}{\varrho_n} + \mathbb{I}_{m} \Big) \boldsymbol{\alpha^{2}}
\end{bmatrix} \Bigg]}_{\phi_{n,1}(\cdot)}.
\end{align*}
First, we claim that $\{\phi_{n,1}(\cdot)\}$ converges uniformly to $\mathbf{K^{T}} \begin{bmatrix} \boldsymbol{\alpha^{1}} \\ \boldsymbol{\alpha^{2}}  \end{bmatrix} $. Consider, 
\begin{align*}
|| \mathbf{K^{T}} \begin{bmatrix} \boldsymbol{\alpha^{1}} \\ \boldsymbol{\alpha^{2}}  \end{bmatrix} - \phi_{n,1}(\boldsymbol{\alpha^{1}};\boldsymbol{\alpha^{2}}) || = \frac{1}{\varrho_n} || \mathbf{K^{T}}  \begin{bmatrix}  \mathbf{K^{1}}\boldsymbol{\alpha^{1}} \\  \mathbf{K^{2}}\boldsymbol{\alpha^{2}}  \end{bmatrix}|| &\geq \frac{1}{\varrho_{n+1}} || \mathbf{K^{T}}  \begin{bmatrix}  \mathbf{K^{1}}\boldsymbol{\alpha^{1}} \\  \mathbf{K^{2}}\boldsymbol{\alpha^{2}}  \end{bmatrix}||\\
 &= || \mathbf{K^{T}} \begin{bmatrix} \boldsymbol{\alpha^{1}} \\ \boldsymbol{\alpha^{2}}  \end{bmatrix} - \phi_{n+1,1}(\boldsymbol{\alpha^{1}};\boldsymbol{\alpha^{2}}) ||. 
\end{align*}
Thus, $\{\bar{\phi}_{n,1}(\boldsymbol{\alpha^{1}};\boldsymbol{\alpha^{2}}) =||\mathbf{K^{T}} \begin{bmatrix} \boldsymbol{\alpha^{1}} \\ \boldsymbol{\alpha^{2}}  \end{bmatrix} - \phi_{n,1}(\boldsymbol{\alpha^{1}};\boldsymbol{\alpha^{2}}) ||\}$ is monotone decreasing sequence. Let $E_{n, \epsilon} = \{(\boldsymbol{\alpha^{1}},\boldsymbol{\alpha^{2}}) \in E: \bar{\phi}_{n,1}(\boldsymbol{\alpha^{1}};\boldsymbol{\alpha^{2}}) < \epsilon\}, \epsilon >0$. Since $\bar{\phi}_{n,1}(\cdot)$ is continuous, $E_{n, \epsilon}$ is open. $\{E_{n, \epsilon}\}$ is an ascending sequence of open sets, i.e., $E_{n} \subset E_{n+1}$. Since $\{\bar{\phi}_{n,1}(\cdot)$ converges pointwise to zero, $\{E_{n, \epsilon}\}$ is an open cover for $E$. Since $E$ is compact, there exists a finite subcover. Since $\{E_{n, \epsilon}\}$ is ascending, the maximum index from the finite subcover is a cover too, i.e., $\exists N_{\epsilon}$ such that $E_{N_{\epsilon}} = E$. We note that $\{\phi_{n,1}(\cdot)\}$ is uniformly bounded, i.e., $\exists c_{\phi_{n,1}}$ such that, $|| \phi_{n,1}(\boldsymbol{\alpha^{1}};\boldsymbol{\alpha^{2}}) || < c_{\phi_{n,1}} \forall  (\boldsymbol{\alpha^{1}};\boldsymbol{\alpha^{2}}) \in E, \forall n$. \\
Second, assuming $\mathbf{K}$ is invertible,  we note that $\{\phi_{n,2}(\cdot)\}$ converges uniformly to $\mathbf{K^{-1}}$ since it is independent of $(\boldsymbol{\alpha^{1}};\boldsymbol{\alpha^{2}})$. Further it is uniformly bounded by $\mathbf{K^{-1}}= c_{\phi_{n,2}}$.\\
Third, we claim that $\{\phi_{n,2}(\cdot)\phi_{n,1}(\cdot)\}$ uniformly converges to $\begin{bmatrix} \boldsymbol{\alpha^{1}} \\ \boldsymbol{\alpha^{2}}  \end{bmatrix}$. Given $\epsilon >0$, there exists $N_{\epsilon}$ such that $|| \phi_{n,1}(\boldsymbol{\alpha^{1}};\boldsymbol{\alpha^{2}}) - \mathbf{K^{T}} \begin{bmatrix} \boldsymbol{\alpha^{1}} \\ \boldsymbol{\alpha^{2}}  \end{bmatrix} || < \frac{\epsilon}{2 \times c_{\phi_{n,2}}}, \forall n \geq N_{\epsilon}, \forall (\boldsymbol{\alpha^{1}};\boldsymbol{\alpha^{2}}) \in E $ and $|| \phi_{n,2}(\boldsymbol{\alpha^{1}};\boldsymbol{\alpha^{2}}) - \mathbf{K^{-1}} || < \frac{\epsilon}{2 \times c_{\phi_{n,1}}}, \forall n \geq N_{\epsilon}$. Thus for $n \geq N_{\epsilon}$,
\begin{align*}
&|| \phi_{n,2}(\boldsymbol{\alpha^{1}};\boldsymbol{\alpha^{2}})\phi_{n,1}(\boldsymbol{\alpha^{1}};\boldsymbol{\alpha^{2}}) -  \begin{bmatrix} \boldsymbol{\alpha^{1}} \\ \boldsymbol{\alpha^{2}}  \end{bmatrix} || = || \phi_{n,2}(\boldsymbol{\alpha^{1}};\boldsymbol{\alpha^{2}})\phi_{n,1}(\boldsymbol{\alpha^{1}};\boldsymbol{\alpha^{2}}) - \mathbf{K^{-1}} \phi_{n,1}(\boldsymbol{\alpha^{1}};\boldsymbol{\alpha^{2}}) + \\
&\mathbf{K^{-1}} \phi_{n,1}(\boldsymbol{\alpha^{1}};\boldsymbol{\alpha^{2}}) - \mathbf{K^{-1}} \mathbf{K^{T}} \begin{bmatrix} \boldsymbol{\alpha^{1}} \\ \boldsymbol{\alpha^{2}}  \end{bmatrix} || \leq || \phi_{n,1}(\boldsymbol{\alpha^{1}};\boldsymbol{\alpha^{2}}) || \; ||  \phi_{n,2}(\boldsymbol{\alpha^{1}};\boldsymbol{\alpha^{2}}) -  \mathbf{K^{-1}} || + ||  \mathbf{K^{-1}}||\\
&|| \phi_{n,1}(\boldsymbol{\alpha^{1}};\boldsymbol{\alpha^{2}}) -  \mathbf{K^{T}} \begin{bmatrix} \boldsymbol{\alpha^{1}}  \\ \boldsymbol{\alpha^{2}}  \end{bmatrix}  || < c_{\phi_{n,1}} \times \frac{\epsilon}{2 \times c_{\phi_{n,1}}} + c_{\phi_{n,2}} \times \frac{\epsilon}{2 \times c_{\phi_{n,2}}} = \epsilon, \forall  (\boldsymbol{\alpha^{1}};\boldsymbol{\alpha^{2}}) \in E
\end{align*}
By the above argument, $\{\mathbf{K}\phi_{n,2}(\cdot)\phi_{n,1}(\cdot)\}$, $\{\phi^{T}_{n,2}(\cdot)\mathbf{K}\phi_{n,2}(\cdot)\phi_{n,1}(\cdot)\}$ and $\{\phi^{T}_{n,1}(\cdot)\phi^{T}_{n,2}(\cdot)\mathbf{K}\phi_{n,2}(\cdot) \\ \phi_{n,1}(\cdot)\}$are uniformly bounded, uniformly converging sequences of continuous functions. Thus, the sequence of continuous functions $\{\phi_{n}(\cdot, \cdot)\}$ converges uniformly to  the continuous function $ \phi(\boldsymbol{\alpha^{1}},\boldsymbol{\alpha^{2}}) =\boldsymbol{\alpha^T}\mathbf{K} \boldsymbol{\alpha}$ where $\boldsymbol{\alpha}= (\boldsymbol{\alpha^{1}},\boldsymbol{\alpha^{2}})$. Thus, swapping limits and supremum, we get
\begin{align*}
\underset{n \to \infty} \lim || T(\varrho_{n}) ||^2 =& \underset{n \to \infty} \lim \underset{(\boldsymbol{\alpha^{1}};\boldsymbol{\alpha^{2}}) \in E} \sup \;   \phi_{n}(\boldsymbol{\alpha^{1}};\boldsymbol{\alpha^{2}}) =
\underset{(\boldsymbol{\alpha^{1}};\boldsymbol{\alpha^{2}}) \in E} \sup  \underset{n \to \infty} \lim  \phi_{n}(\boldsymbol{\alpha^{1}};\boldsymbol{\alpha^{2}}) \\
= & \underset{(\boldsymbol{\alpha^{1}};\boldsymbol{\alpha^{2}}) \in E} \sup \; \boldsymbol{\alpha^T}\mathbf{K} \boldsymbol{\alpha} \leq \lambda_{\text{max}}(\mathbf{K}).
\end{align*}
$\lambda_{\text{max}}(\mathbf{K}) \leq 1$ would be sufficient for the statement in the  current proposition to hold. However analyzing the blocks of $\mathbf{K}$ provides some insights,
\begin{align*}
&\mathbf{K} = \hspace{-4pt}
\begin{bmatrix} 
&\hspace{-4pt}\mathbf{\tilde{K}^{1}} &\mathbf{K^{12}} \\ 
&\hspace{-4pt}\mathbf{K^{{12}^T}} &\mathbf{\tilde{K}^{2}} 
\end{bmatrix} \hspace{-4pt}=\hspace{-5pt}
\underbrace{\begin{bmatrix} 
&\mathbb{I}_{m} &\mathbf{K^{12}}\mathbf{\tilde{K}^{2^{-1}}} \\ 
&0 &\mathbb{I}_{m}
\end{bmatrix}}_{P^T}\hspace{-4pt}
\underbrace{\begin{bmatrix}
&\mathbf{\tilde{K}^{1}}  - \mathbf{K^{12}}\mathbf{\tilde{K}^{2^{-1}}}\mathbf{K^{{12}^T}} &0\\
&0 &\mathbf{\tilde{K}^{2}} 
\end{bmatrix}}_{D}\hspace{-4pt}
\underbrace{\begin{bmatrix}
&\mathbb{I}_{m} &\mathbf{K^{12}}\mathbf{\tilde{K}^{2^{-1}}} \\ 
&0 &\mathbb{I}_{m}
\end{bmatrix}^{T}}_{P}
\end{align*}
Invoking \textit{Schur's complement}, we get the above decomposition of $\mathbf{K}$. Thus, 
\begin{align*}
\lambda_{\text{max}}(\mathbf{K}) =\underset{x \neq 0} \sup \;  \frac{x^{T}\mathbf{K}x}{x^Tx} = \underset{x \neq 0} \sup \;  \frac{x^{T}P^T D Px}{x^Tx} \overset{(a)}{=} \underset{y \neq 0} \sup \;  \frac{y^{T} D y}{y^Ty} = \lambda_{\text{max}}(D) \overset{(b)}{\leq} \lambda_{\text{max}}(\mathbf{\tilde{K}^{1}}),\lambda_{\text{max}}(\mathbf{\tilde{K}^{2}}).
\end{align*}
Equality $(a)$ is true because $P$ is a full rank matrix ($\mathcal{R}(P) = \mathbb{R}^{2m}$) with all eigenvalues equal to $1$. Note that $\mathbf{K^{12}}\mathbf{\tilde{K}^{2^{-1}}}\mathbf{K^{{12}^T}}$ is positive definite. Thus $\lambda_{\max}(\mathbf{\tilde{K}^{1}}  - \mathbf{K^{12}}\mathbf{\tilde{K}^{2^{-1}}}\mathbf{K^{{12}^T}}) \leq \lambda_{\max}(\mathbf{\tilde{K}^{1}})$. The set of eigenvalues of $D$ is the union of the set of eigenvalues of $\mathbf{\tilde{K}^{1}}  - \mathbf{K^{12}}\mathbf{\tilde{K}^{2^{-1}}}\mathbf{K^{{12}^T}}$ and $\mathbf{\tilde{K}^{2}}$.  Hence, $\lambda_{\text{max}}(D) \leq \lambda_{\text{max}}(\mathbf{\tilde{K}^{1}} - \mathbf{K^{12}}\mathbf{\tilde{K}^{2^{-1}}}\mathbf{K^{{12}^T}}),\lambda_{\text{max}}(\mathbf{\tilde{K}^{2}})$. Hence, inequality $(b)$ follows. Thus, $\underset{n \to \infty} \lim || T(\varrho_{n}) ||^2 \leq \max(\lambda_{\text{max}}(\mathbf{\tilde{K}^{1}}),\lambda_{\text{max}}(\mathbf{\tilde{K}^{2}})) \leq 1$
\end{proof}
It is not necessary that $\lambda_{\max}(\mathbf{K}) \leq 1 $ or $\max(\lambda_{\max}(\mathbf{\tilde{K}^{1}}), \lambda_{\max}(\mathbf{\tilde{K}^{2}})) \leq 1$. However, by suitably normalizing $\mathbf{K}$ or $\mathbf{\tilde{K}^{1}}, \mathbf{\tilde{K}^{2}}$ matrices for large $n$, we can ensure that $\lambda_{\text{max}}(\mathbf{K}) \leq 1$ and thus guarantee that $|| T(\varrho_{n_{k_l}})||^2 $ converges to a value less than or equal to $1$.
\subsection{Consistency of the Learning Algorithm}\label{Subsection 3.3}
We define the learning operator for the multi-agent system for each iteration which takes into account learning at the agents, the uploading operation, the fusion operation at the fusion center and the downloading operation on to the agents. Then, we go on to define the learning operator for the system for an execution of the learning algorithm and prove consistency properties of Algorithm \ref{Algorithm 1}.
\begin{definition}\label{Definition 8}
The multi-agent learning and upload operator, $\bar{T}$, is defined as $\bar{T}: H^{1} \times H^{2} \times H^{1} \times H^{2}  \to H \times H$ defined as
\begin{align*}
\bar{T}(\varrho^{1};\varrho^{2})[f^{1}; f^{2};\psi^{1}_{(x^{1},y^{1})}, \psi^{2}_{(x^{2},y^{2})}] =  \Big[\hat{L}^{1} \circ \bar{T}^{1}(\varrho^{1}) [  f^{1};  \psi^{1}_{(x^{1},y^{1})}]; \hat{L}^{2} \circ \bar{T}^{2}(\varrho^{2})[f^{2}; \psi^{2}_{(x^{2},y^{2})}]\Big].
\end{align*}
\end{definition}
\begin{proposition}
Let $\{\varrho^{1}_{n}\}$ and $\{\varrho^{2}_{n}\}$ be sequences which diverge to $\infty$. Then, there exists  a subsequence of $\{\bar{T}(\varrho^{1}_{n}, \varrho^2_n)\}$, $\{\bar{T}(\varrho^{1}_{n_k}, \varrho^2_{n_K})\}$ such that $\underset{k \to \infty} \lim || \bar{T}(\varrho^{1}_{n_k}, \varrho^2_{n_K})|| = 1 $  
\end{proposition}
\begin{proof}
The norm of the multi-agent learning operator can be computed as,
\begin{align*}
|| \bar{T}(\varrho^{1}_{n};\varrho^{2}_{n})||^{2} &= \hspace{-10pt}\underset{(f^{1};\psi^{1}_{(x^{1},y^{1})};f^{2};\psi^{2}_{(x^{2},y^{2})}) \neq \theta} \sup \frac{\sum_{i= 1, 2}|| \hat{L}^{i} \circ \bar{T}^{i}(\varrho^{i}_{n})[f^{i},\psi^{i}_{(x^{i},y^{i})}  ] ||^2}{ || f^{1} ||^{2}_{H^{1}}  \hspace{-3pt} + \hspace{-3pt}|| \psi^{1}_{(x^{1},y^{1})} ||^{2}_{H^1}  \hspace{-3pt} +  \hspace{-3pt} || f^{2} ||^{2}_{H^{2}}  \hspace{-3pt} +  \hspace{-3pt} || \psi^{2}_{(x^{2},y^2)} ||^{2}_{H^2}} \\
& \overset{(a)}{=} \hspace{-17pt}\underset{(f^{1};\psi^{1}_{(x^{1},y^{1})};f^{2};\psi^{2}_{(x^{2},y^{2})}) \neq \theta} \sup \frac{\sum_{i= 1, 2}|| \hat{L}^{i} ||\; || \bar{T}^{i}(\varrho^{i}_{n})|| \; (||f^{i}||^2 + ||\psi^{i}_{(x^{i},y^{i})}||^2)}{ || f^{1} ||^{2}_{H^{1}}  \hspace{-3pt} + \hspace{-3pt}|| \psi^{1}_{(x^{1},y^{1})} ||^{2}_{H^1}  \hspace{-3pt} +  \hspace{-3pt} || f^{2} ||^{2}_{H^{2}}  \hspace{-3pt} +  \hspace{-3pt} || \psi^{2}_{(x^{2},y^2)} ||^{2}_{H^2}}\\
\underset{k \to \infty} \lim \hspace{5pt}|| \bar{T}(\varrho^{1}_{n_k}; \varrho^{2}_{n_k} )||^{2} &= \hspace{-15pt} \underset{(\{f^{i};\psi^{i}_{(x^{i},y^{i})}\}_{i=1,2}) \neq \theta}\sup \; \hspace{-3pt} \underset{k \to \infty} \lim \frac{\sum_{i= 1, 2} || \hat{L}^{i} ||\; || \bar{T}^{i}(\varrho^{i}_{n})|| \; (||f^{i}||^2 + ||\psi^{i}_{(x^{i},y^{i})}||^2)}{ \sum_{i=1,2} || f^{i} ||^{2}_{H^{i}}  + || \psi^{i}_{(x^{i},y^{i})} ||^{2}_{H^i}}\\
&\overset{(b)}{=} \frac{\sum_{i= 1, 2} (||f^{i}||^2 + ||\psi^{i}_{(x^{i},y^{i})}||^2)}{ \sum_{i=1,2} || f^{i} ||^{2}_{H^{i}}  + || \psi^{i}_{(x^{i},y^{i})} ||^{2}_{H^i}} =1,
\end{align*}
where equality $(a)$ is achieved by choosing a sequence $\{f^{i}_{k}, \psi^{i}_{x^{i}_{k}, y^i_{k}}\} \subset H^{i}$ such that the norm is achieved; in case such that a sequence does not exist then equality $(a)$ would be replaced with strict inequality; $(b)$ follows from Proposition \ref{Proposition 4} and Corollary \ref{Corollary 19}. 
\end{proof}
We note that mapping from $(x,y) \mapsto \psi^{i}_{(x,y)}$ is a nonlinear map however bounded under suitable assumptions on $K$. This mapping has not been incorporated into the definition of $\bar{T}$ to avoid complications associated with nonlinearity. 
\begin{definition}
The multi-agent download operator, $\hat{T}: H \to H^{1} \times H^{2}$, is defined as 
\begin{align*}
\hat{T}(f): \frac{1}{c_d} \Big[\sqrt{\bar{L}^{1}} \circ \Pi_{\mathcal{N}\big(\sqrt{\bar{L}^1}\big)^{\perp}}\big(f\big) ; \sqrt{\bar{L}^{2}} \circ \Pi_{\mathcal{N}\big(\sqrt{\bar{L}^2}\big)^{\perp}}\big(f\big)  \Big] 
\end{align*}
where $\bar{L}^{i}$ is defined in Lemma \ref{Lemma 22} and $c_{d} = \underset{(f \neq \theta)} \sup \; 1 + \frac{\langle f, \Pi_{\mathcal{N}\big(\sqrt{\bar{L}^1}\big)^{\perp}}\Pi_{\mathcal{N}\big(\sqrt{\bar{L}^2}\big)^{\perp}}\big(f \big) \rangle_{H}}{ || f||^{2}_H}$. The norm of $\hat{T}$, $||\hat{T} ||$ is equal to 1.
\end{definition}
The \textit{unnormalized} norm of the download operator is,
\begin{align*}
|| \hat{T} ||^{2} &= \underset{(f \neq \theta)} \sup \; \frac{ || \sqrt{\bar{L}^{1}} \circ \Pi_{\mathcal{N}\big(\sqrt{\bar{L}^1}\big)^{\perp}}\big(f\big) ||^{2}_{H^1} + || \sqrt{\bar{L}^{2}} \circ \Pi_{\mathcal{N}\big(\sqrt{\bar{L}^2}\big)^{\perp}}\big(f\big) ||^{2}_{H^2}}{ || f ||^{2}_{H}},\\
&=  \underset{(f \neq \theta)} \sup \frac{||\Pi_{\mathcal{N}\big(\sqrt{\bar{L}^1}\big)^{\perp}}\big(f\big) ||^{2}_{H} + || \Pi_{\mathcal{N}\big(\sqrt{\bar{L}^2}\big)^{\perp}}\big(f\big) ||^{2}_{H}}{ || f ||^{2}_{H}}, \text{equality invoking Theorem \ref{Theorem 23}} ,
\end{align*}
which is equal to $c_d$. We note that the normalization of the download operator is not needed for all $n$. For $n$ sufficiently large, we normalize the downloaded functions by $c_d$ to so that $|| \hat{T} || = 1$. This done to ensure that the learning operator (defined next) has norm converging to $1$. To summarize (as in Figure \ref{Figure 2}), the operator $\bar{T}(\varrho^{1}_n, \varrho^2_n)$ takes in $(\bar{f}^{1}_{n-1}; \psi^{2}_{x^1_n, y^1_n}, \bar{f}^{2}_{n-1}; \psi^{2}_{x^2_n, y^2_n})$ and outputs $f^{1}_n, f^{2}_n$, operator $T(\varrho_n)$ takes in $f^{1}_n$ and $f^2_n$ and outputs $f_{n}$, operator $\hat{T}$ takes in $f_{n}$ and outputs $\bar{f}^{1}_n$ and $\bar{f}^{2}_n$. Thus, the learning operator at stage $n$ is defined as the composition of the operators, $\bar{T}(\varrho^{1}_n, \varrho^2_n)[\cdot]$, $T(\varrho_n)[\cdot]$, and $\hat{T}(\cdot)$, as below.
\begin{figure}
\begin{center}
\includegraphics[scale=0.44]{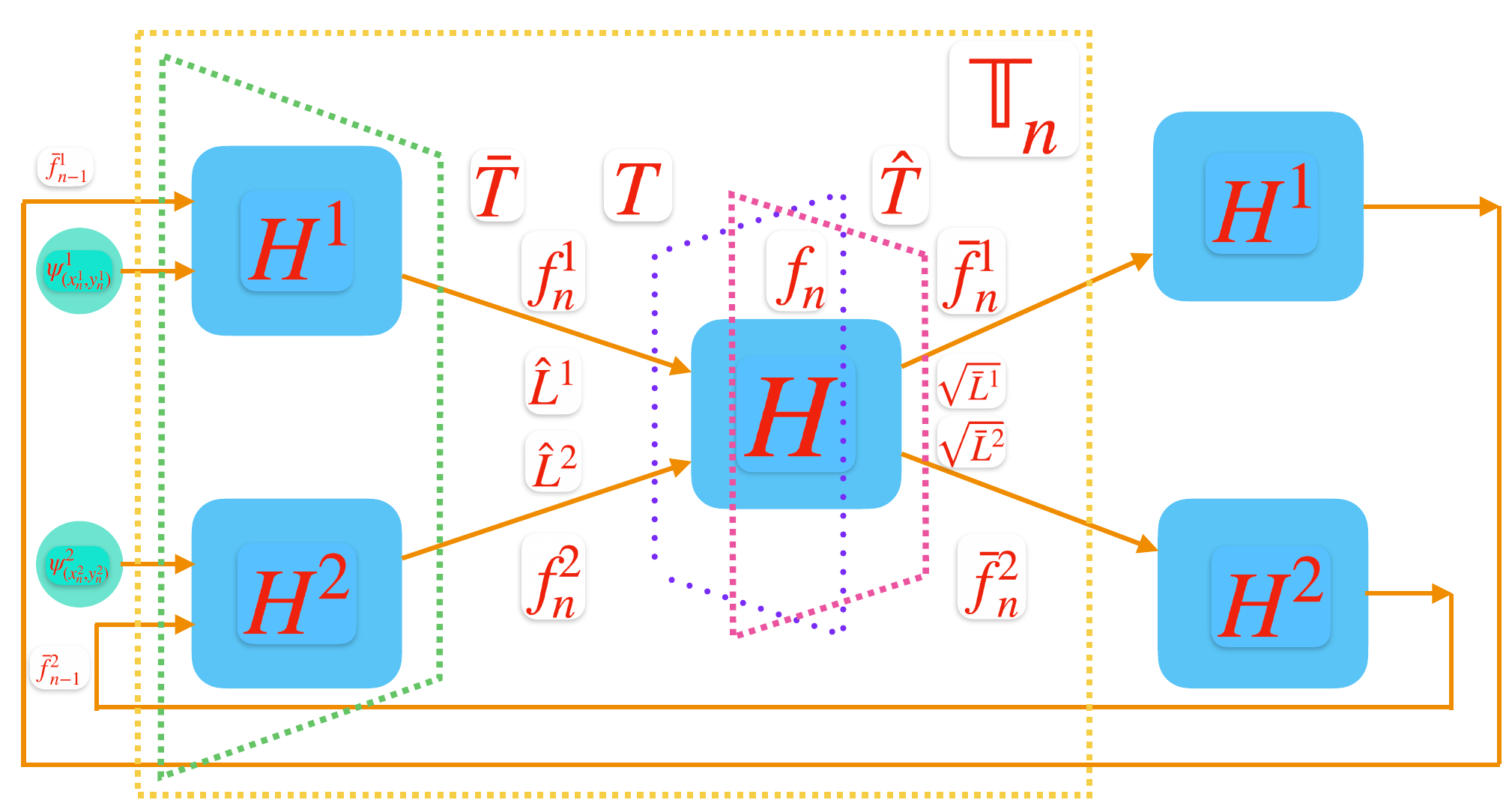}
\caption{The learning operator at stage $n$, $\mathbb{T}_{n}$, obtained through composition of operators, $\bar{T}, T, \hat{T}.$}\label{Figure 2}
\end{center}
\end{figure}
\begin{definition}\label{Definition 9}
The learning operator at stage $n$ is defined as $ \mathbb{T}_{n}: H^{1} \times H^{2} \times H^{1} \times H^{2}\to H^{1} \times H^{2}$, 
\begin{align*}
 \mathbb{T}_{n}(\varrho^{1}_{n}, \varrho^{2}_{n}, \varrho_n)[f^{1};f^{2};\psi^{1}_{x^{1},y^{1}}; \psi^{2}_{x^{2},y^{2}}] = \hat{T} \circ  T(\varrho_{n}) \circ \bar{T}(\varrho^{1}_{n}, \varrho^{2}_{n})[f^{1};f^{2}; \psi^{1}_{x^{1},y^{1}}; \psi^{2}_{x^{2},y^{2}}].  
\end{align*}
\end{definition}
\begin{proposition}\label{Proposition 10}
Let $\{\varrho^{1}_{n}\}$, $\{\varrho^{2}_{n}\}$ and $\{\varrho_{n}\}$ be sequences which diverge to $\infty$.  There exists a subsequence of $\{ \mathbb{T}_{n}\}$, $ \{\mathbb{T}_{n_{k_l}}\}$ such that $\underset{ l \to \infty} \lim \; \prod^{l}_{p=1} || \mathbb{T}_{n_{k_p}} || = c_{M_1} < \infty$ and $\underset{ l \in \mathbb{N}} \sup \; \prod^{l}_{p=1} || \mathbb{T}_{n_{k_p}} || = c_{M_2} < \infty$.  
\end{proposition}
\begin{proof}
Let $a_{n} = || \mathbb{T}_{n}||$. Then, $0 \leq a_{n} \leq || \hat{T} || \times || T(\varrho_n)  || \times ||\bar{T}_{n}(\{(x^{i}_{n};y^{i}_{n};\varrho^{i}_{n})\}_{i=1,2}) ||$. From Propositions \ref{Proposition 4} and \ref{Proposition 7}, there is subsequence $\{a_{n_k}\}$ that converges to a positive constant, $c_5$, less than or equal to 1. W.L.O.G we assume that the subsequence converges to 1 as the sequence $\{\frac{a_n}{c_5}\}$ converges to 1. There is a subsequence  of $\{a_{n_k}\}$, $\{a_{n_{k_l}}\}$ such that $ | a_{n_{k_l}} -1 | \leq \frac{1}{2^l}$. We define the sequence $b_{n} = \prod^{n}_{l=1}a_{n_{k_l}}$. We claim that the sequence $\{b_{n}\}$ converges. Indeed, consider $\ln(b_{n}) =\sum^{n}_{l=1}\ln(a_{n_{k_l}})$. Since  $\frac{-1}{2^l-1} \leq \ln(a_{n_{k_l}}) \leq \frac{1}{2^l}$, $|\ln(a_{n_{k_l}})| \leq \frac{1}{2^{l}-1}$. Since $\sum^{\infty}_{l=1}\frac{1}{2^{l}-1}$ converges, for every $\epsilon > 0$, $\exists N_{\epsilon}$ such that $\sum^{\infty}_{l=N_{e}} |\ln(a_{n_{k_l}})| < \epsilon$. Thus, $|\ln(b_{p}) - \ln(b_{q})| \leq \sum^{p}_{l=q} |\ln(a_{n_{k_l}})| < \epsilon \forall p, q \geq N_{\epsilon}$. Thus, the sequence $\{\ln(b_n)\}$ is cauchy and converges to a real number. By continuity of the natural logarithm function, sequence $\{b_{n}\}$ converges to a real number lying in the interval $\big[\prod^{\infty}_{l=1} \frac{2^l-1}{2^l}, \prod^{\infty}_{l=1} \frac{2^l+1}{2^l}\big]$. From the construction of the subsequence $\{a_{n_{k_l}}\}$, it follows that $\underset{ l \in \mathbb{N}} \sup \; \prod^{l}_{p=1} || \mathbb{T}_{n_{k_p}} || \leq  \prod^{\infty}_{l=1} \frac{2^l+1}{2^l} < \infty$.
\end{proof}
The groundwork needed to define the learning operator given initial estimates, $f^{1}, f^{2}$ and a sequence of data points $(\{(x^{1}_n; y^{1}_{n})\}, \{(x^{2}_n; y^{2}_{n})\})$ at the agents has been laid. Consider the set of operators $\{\mathbb{T}_{n}\}^{n_{k_1}-1}_{n=1} \cup \{\mathbb{T}_{n_{k_l}}\}^{\infty}_{l=1}$, where the sequence of operators $\{\mathbb{T}_{n_{k_l}}\}^{\infty}_{l=1}$ is such that $\underset{l \in \mathbb{N}} \sup \prod^{l}_{p=1} || \mathbb{T}_{n_{k_l}} ||$ is finite (Proposition  \ref{Proposition 10}). We re-index the set to obtain the countable collection of operators $\{\mathbb{T}_{n}\}^{\infty}_{n=1}$ for which $\underset{n \in \mathbb{N}} \sup \prod^{n}_{l=1} || \mathbb{T}_{l} || = c_{M} < \infty$.
\begin{definition}\label{Definition 11}
Given $(f^{1}_0;f^{2}_0)$, $\{\varrho^{1}_n, \varrho^{2}_n, \varrho_n\}$, and sequence of data points $\{(x^{1}_n; y^{1}_{n})\}, \{(x^{2}_n;y^{2}_{n}\})$, the learning operator, $\mathbb{T}:  H^{1} \times H^{2} \times H^{1} \times H^{2} \to H^{1} \times H^{2}$ is recursively  defined as follows
\begin{align*}
\bar{\mathbb{T}}_{n}(\{(\varrho^{1}_k;\varrho^{2}_k; \varrho_k)\}^{n}_{k=1})\Big[f^{1}_{0},f^{2}_{0};\psi^{1}_{x^{1}_1,y^{1}_1};& \psi^{2}_{x^{2}_1,y^{2}_1}\Big] = \mathbb{T}_{n}(\varrho^{1}_{n}, \varrho^{2}_{n}, \varrho_n)\Big[\bar{\mathbb{T}}_{n-1}(\{(\varrho^{1}_k;\varrho^{2}_k; \varrho_k)\}^{n-1}_{k=1}\})\\
&\Big[\bar{f}^{1}_{n-2},\bar{f}^{2}_{n-2}, ;\psi^{1}_{x^{1}_{n-1},y^{1}_{n-1}}; \psi^{2}_{x^{2}_{n-1},y^{2}_{n-1}}\Big]\psi^{1}_{x^{1}_{n},y^{1}_{n}}; \psi^{2}_{x^{2}_{n},y^{2}_{n}}\Big]\;,\\
\bar{\mathbb{T}}_{1}(\varrho^{1}_1;\varrho^{2}_1; \varrho_1)\Big[f^{1}_{0},f^{2}_{0};\psi^{1}_{x^{1}_1,y^{1}_1}; \psi^{2}_{x^{2}_1,y^{2}_1}\Big] &=  \mathbb{T}_{1}(\varrho^{1}_1;\varrho^{2}_1; \varrho_1)\Big[f^{1}_{0},f^{2}_{0};\psi^{1}_{x^{1}_1,y^{1}_1}; \psi^{2}_{x^{2}_1,y^{2}_1}\Big], \\
\mathbb{T}(\{(\varrho^{1}_n;\varrho^{2}_n;\varrho_n)\})\Big[f^{1}_0;f^{2}_0;\psi^{1}_{x^{1}_1,y^{1}_1}; \psi^{2}_{x^{2}_1,y^{2}_1}\Big] &= \underset{ n \to \infty } \lim \bar{\mathbb{T}}_{n}(\{(\varrho^{1}_k;\varrho^{2}_k; \varrho_k)\}^{n}_{k=1})\Big[f^{1}_{0},f^{2}_{0};\psi^{1}_{x^{1}_1,y^{1}_1}; \psi^{2}_{x^{2}_1,y^{2}_1}\Big].
\end{align*}
\end{definition}
\begin{proposition}\label{Proposition 12}
Let $\{\varrho^{1}_{n}\}$, $\{\varrho^{2}_{n}\}$ and $\{\varrho_{n}\}$ be sequences which diverge to $\infty$. Then, $\bar{\mathbb{T}}_{n}(\cdot)\Big[f^{1}_{0},f^{2}_{0};\\ \psi^{1}_{x^{1}_1,y^{1}_1}; \psi^{2}_{x^{2}_1,y^{2}_1}\Big] $ is a uniformly equicontinuous sequence of linear operators. 
\end{proposition}
\begin{proof}
Since the composition of linear operators is linear, $\{\mathbb{T}_{n}\}$ and $\{\bar{\mathbb{T}}_{n}\}$ are sequences of linear operators. From Proposition \ref{Proposition 10} , the sequence of operators $\{\bar{\mathbb{T}}_{n}\}$ is uniformly bounded, $||\bar{\mathbb{T}}_{n}|| \leq \prod^{n}_{k=1} || \mathbb{T}_{n}|| \leq  \underset{n \in \mathbb{N}} \sup \prod^{n}_{k=1} || \mathbb{T}_{k} || = c_{M} < \infty $. Given $\epsilon >0$, let $\delta < \frac{\epsilon}{c_{M}}$. Then, $|| f - g|| < \delta $ implies $|| \bar{\mathbb{T}}_{n}(f) - \bar{\mathbb{T}}_{n}(g) || \leq || \bar{\mathbb{T}}_{n} || || f- g || < c_{M} \times \frac{\epsilon}{c_{M}} < \epsilon \forall n \in \mathbb{N}$.
\end{proof}
For the above proof and the following results, we suppress the data related arguments in the definition of operators $\bar{\mathbb{T}}_{n}(\cdot)[\cdot]$ and $\mathbb{T}(\cdot)[\cdot]$ and use the notation $\bar{\mathbb{T}}_{n}(\cdot)$ and $\mathbb{T}(\cdot)$ where the arguments are the functions that they operate on. 
\begin{theorem}\label{Theorem 13}
The learning operator is well defined in the following sense; there exists a subsequence of $\{\bar{\mathbb{T}}_{n}\Big(f^{1}_0;f^{2}_0; \psi^{1}_{x^{1}_1,y^{1}_1}; \psi^{2}_{x^{2}_1,y^{2}_1}\Big)\}$ which strongly converges to $(f^{1,*}; f^{2,*})\in H^{1} \times H^{2}$ for any $(f^{1}_0;f^{2}_0; \psi^{1}_{x^{1}_1,y^{1}_1}; \psi^{2}_{x^{2}_1,y^{2}_1}) \in H^{1} \times H^{2} \times H^{1} \times H^{2}$. $(f^{1,*}; f^{2,*})$ depends on $(f^{1}_0;f^{2}_0; \psi^{1}_{x^{1}_1,y^{1}_1}; \psi^{2}_{x^{2}_1,y^{2}_1})$. $\mathbb{T}$ is a linear and bounded operator.
\end{theorem}
\begin{proof}
Since $H^{1} \times H^{2} \times H^{1} \times H^{2} $ is finite dimensional space, it is separable. Let $\{\psi_n\}$ be an enumeration of a dense set, $\mathbb{D}$, in $H^{1} \times H^{2} \times H^{1} \times H^{2}$. Consider the sequence, $\Psi_{n}(\psi_1)[f] = \langle \bar{\mathbb{T}}_{n}(\psi_1), f \rangle_{H^1 \times H^2}, f \in H^{1} \times H^{2}$. Since $|| \bar{\mathbb{T}}_{n}|| \leq \underset{n \in \mathbb{N}} \sup \prod^{n}_{l=1} || \mathbb{T}_{l} || = c_{M}  < \infty$, by the CBS inequality $ | \langle \bar{\mathbb{T}}_{n}(\psi_1), f \rangle_{H^1 \times H^2} | \leq c_{M} \;  || \psi_1 ||_{H^1 \times H^2 \times H^{1} \times H^{2}} \; || f ||_{H^1 \times H^2}  $ , i.e., the sequence $\{\Psi_{n}(\psi_1)[f]\}$ is bounded sequence of bounded linear functionals on a separable Hilbert space. By Helley's theorem there exists a subsequence, a strictly increasing sequence of integers, $\{s(1,n)\}$, such that $\{\Psi_{s(1,n)}(\psi_1)[\cdot]\}$ converges pointwise to $\Psi^{*}(\psi_1)[\cdot] \in (H^{1} \times H^{2})^*$. By the \textit{Riesz- Fr\'echet Representation Theorem}, $\exists  \psi^{*}_1 \in H^{1} \times H^{2}$ such that $\Psi^{*}(\psi_1)[f] = \langle  \psi^{*}_1, f \rangle, \forall f \in H^{1} \times H^{2}$. Thus, $ \bar{\mathbb{T}}_{s(1,n)}(\psi_1) \rightharpoonup \psi^{*}_1$. Since all the inner products and norms in this proof are evaluated in $H^1 \times H^2$, going forward this is suppressed.  By the same argument, since the sequence $\langle \bar{\mathbb{T}}_{n}(\psi_2), f \rangle$ is bounded, there exists a subsequence of $\{s(1,n)\}$, $\{s(2,n)\}$, such that $\bar{\mathbb{T}}_{s(2,n)}(\psi_2) \rightharpoonup \psi^{*}_2, \psi^{*}_2 \in H^{1} \times H^{2}$. We can inductively continue this process to obtain a strictly increasing sequence of integers $\{s(j,n)\}$ which is a subsequence of $\{s(j-1,n)\}$ such that $\bar{\mathbb{T}}_{s(j,n)}(\psi_k) \rightharpoonup \psi^{*}_k$. For each $j$, we define $\mathbb{T}(\psi_j)$ as $\psi^{*}_j$. By \textit{Cantor's diagonalization argument}, consider the subsequence of operators, $\bar{\mathbb{T}}_{n_k}(\cdot)$ , where $n_{k} =s(k,k)$ for every index $k$. For each $j$, the subsequence, $\{n_{k}\}^{\infty}_{k=j}$ is a subsequence of the $j$th subsequence of integers chosen before and thus $ \bar{\mathbb{T}}_{n_k}(\psi_{j}) \overset{k}{\rightharpoonup} \mathbb{T}(\psi_j)\; \forall j \in \mathbb{N}.$ Thus, $\{\bar{\mathbb{T}}_{n_k}(f) \}$ converges weakly to  $\mathbb{T}(f)$ on $\mathbb{D}$. Let $g$ be any function in $H^{1} \times H^{2} \times H^{1} \times H^{2}$. We claim that the sequence $\{\Psi_{n_k}(g)[f]\}$ is Cauchy for every $f$, where $ \Psi_{n_k}(g)[f] = \langle \bar{\mathbb{T}}_{n_k}(g), f \rangle$. Indeed, since the sequence operators $\{\bar{\mathbb{T}}_{n_k}\}$ is equicontinuous, $\forall\epsilon >0, \exists  \delta > 0$ such that, $|| \bar{\mathbb{T}}_{n_k}(h) - \bar{\mathbb{T}}_{n_k}(g)|| < \frac{\epsilon}{3 \times ||f||}$ for all $h$ such that $|| h -g || <\delta$ and all indices $n_k$. This implies that $ | \langle \bar{\mathbb{T}}_{n_k}(h), f \rangle  - \langle \bar{\mathbb{T}}_{n_k}(g), f \rangle | \leq || \bar{\mathbb{T}}_{n_k}(h) - \bar{\mathbb{T}}_{n_k}(g)||\; ||f|| <  \frac{\epsilon}{3}$ $\forall h$ such that $|| h -g || <\delta$ and all indices $n_k$. Since $\mathbb{D}$ is dense in $H^{1} \times H^{2} \times H^{1} \times H^{2}$, there exists $\psi \in \mathbb{D}$ such that $||\psi - g|| < \delta$. Since the sequence, $\{ \langle \bar{\mathbb{T}}_{n}(\psi), f \rangle\}$ is cauchy, there exists $N_{\epsilon}$ such that $|\langle \bar{\mathbb{T}}_{n_p}(h), f \rangle  - \langle \bar{\mathbb{T}}_{n_q}(g), f \rangle  | < \frac{\epsilon}{3}, \; \forall p, q \geq N_{\epsilon}$. Thus, for all $p,q \geq N_{\epsilon}$
\begin{align*}
|\Psi_{n_p}(g)[f] - \Psi_{n_q}(g)[f]| \leq |\Psi_{n_p}(g)[f] - \Psi_{n_p}(\psi)[f]| + |\Psi_{n_p}(\psi)[f] - \Psi_{n_q}(\psi)[f]|+ \\ |\Psi_{n_q}(\psi)[f] - \Psi_{n_q}(g)[f]| < \frac{\epsilon}{3} + \frac{\epsilon}{3} + \frac{\epsilon}{3} = \epsilon.
\end{align*}
Thus, $\{\Psi_{n_k}(g)[f]\}$ converges to real number which denote by $\Psi^*(g)[f]$. From the linearity of each functional in the sequence $\{\Psi_{n_k}(g)[\cdot]\}$, it follows that $\Psi^*(g)[\cdot]$ is linear. Since, $| \Psi_{n_k}(g)[f] | \leq c_M ||g||\; ||f||, \forall n_{k}$, it follows that $\underset{k \to \infty} \lim  | \Psi_{n_k}(g)[f] | = |\underset{k \to \infty} \lim   \Psi_{n_k}(g)[f]| = \hspace{-3pt} |\Psi^*(g)[f]|  \leq c_M ||g||\; ||f||.$ Thus, $\Psi^*(g)[\cdot] \in (H^{1} \times H^{2})^*$. By the \textit{Riesz- Fr\'echet Representation Theorem}, $\exists  \psi^{*}_g \in H^{1} \times H^{2}$ such that $\Psi^{*}(g)[f] = \langle  \psi^{*}_g, f \rangle, \forall f \in H^{1} \times H^{2}$. Thus, $\langle \bar{\mathbb{T}}_{n_k}(g), f \rangle \to \langle \psi^{*}_g, f \rangle, \forall f \in H^{1} \times H^{2}$.  We let $\mathbb{T}(g) = \psi^{*}_g$ and from the arguments presented we concluded that $\bar{\mathbb{T}}_{n_k}(g) \rightharpoonup \mathbb{T}(g) \forall g \in H^{1} \times H^{2}$. Let $\{\hat{\varphi}_{j}\}^{|\mathcal{I}^1| + |\mathcal{I}^2|}_{j=1} \hspace{-3pt}$ be an orthonormal basis for $H^{1} \times H^{2}$ obtained by applying the \textit{Gram–Schmidt orthonormalization process} to $\{\varphi^{1}_{j} \times \theta^2 \}_{j \in \mathcal{I}^{1}} \cup \{\theta^1 \times \varphi^{2}_{j} \}_{j \in \mathcal{I}^{2}}$. Consider the linear functionals, $\Psi_j(\hat{\varphi}_{j})[f]  = \langle \hat{\varphi}_{j} , f \rangle, f \in H^{1} \times H^{2}, j=1, \ldots, |\mathcal{I}^1| + |\mathcal{I}^2|$. By the weak convergence result, $\Psi_j(\hat{\varphi}_{j})[\bar{\mathbb{T}}_{n_k}(g)] \to \Psi_j(\hat{\varphi}_{j})[\mathbb{T}(g)] \forall g \in H^{1} \times H^{2}$ and $j$. This implies that, $\forall \epsilon >0$, $\exists N_{\epsilon}$ such that $|\Psi_j(\hat{\varphi}_{j})[\bar{\mathbb{T}}_{n_k}(g)] - \Psi_j(\hat{\varphi}_{j})[\mathbb{T}(g)] | < \frac{\epsilon}{|\mathcal{I}^1| + |\mathcal{I}^2|}, \forall j, \forall k \geq N_{\epsilon}$. Thus, $\forall \epsilon >0$, $\exists N_{\epsilon}$ such that 
\begin{align*}
|| \bar{\mathbb{T}}_{n_k}(g) - \mathbb{T}(g) || &= || \sum^{|\mathcal{I}^1| + |\mathcal{I}^2|}_{j=1} \Big [ \Psi_j(\hat{\varphi}_{j})[\bar{\mathbb{T}}_{n_k}(g)] - \Psi_j(\hat{\varphi}_{j})[\mathbb{T}(g)] \Big] \hat{\varphi}_j || \\
&\leq \hspace{-3pt} \sum^{|\mathcal{I}^1| + |\mathcal{I}^2|}_{j=1} \hspace{-3pt}  | \Psi_j(\hat{\varphi}_{j})[\bar{\mathbb{T}}_{n_k}(g)] - \Psi_j(\hat{\varphi}_{j})[\mathbb{T}(g)] | < (|\mathcal{I}^1| + |\mathcal{I}^2|) \times \frac{\epsilon}{|\mathcal{I}^1| + |\mathcal{I}^2|} = \epsilon.
\end{align*}
Thus, $\{ \bar{\mathbb{T}}_{n_k}(g) \}$ converges to $\mathbb{T}(g)$ in strong topology for all $g \in H^{1} \times H^{2}$. From the linearity of the sequence $\{ \bar{\mathbb{T}}_{n_k} \}$, it follows that $\mathbb{T}$ is linear. $|| \mathbb{T}(f) || = \underset{k \to \infty}  \lim ||\bar{\mathbb{T}}_{n_k}(f) || \leq \underset{k \to \infty} \lim  ||\bar{\mathbb{T}}_{n_k} || \; || (f) || \leq c_{M} ||f ||$. Thus,  $\mathbb{T}$ is bounded.
\end{proof}
\begin{remark}\label{Remark 14}
We note that weak convergence result in the proof of theorem holds for seperable Hilbert spaces, not just finite dimensional or compact Hilbert spaces. The proof of extending weak convergence to strong convergence holds only for finite dimensional spaces.
\end{remark}
\begin{corollary}\label{Corollary 15}
The learning algorithm presented in algorithm 1 is strongly consistent. 
\end{corollary}
\begin{proof}
Follows from Definition \ref{Definition 1} and Theorem \ref{Theorem 13}. 
\end{proof}
\begin{lemma}\label{Lemma 16}
$\mathbb{T}(\cdot)$ is continuous. 
\end{lemma}
\begin{proof}
Since $\mathbb{T}(\cdot)$ is linear and bounded it continuous. Alternatively it can be proven as follows. Since $\{\bar{\mathbb{T}}_{n_k}\}$ is a uniformly equicontinuous sequence of operators, given $\epsilon > 0, \exists \delta> 0$, such that $|| f - g || < \delta \implies || \bar{\mathbb{T}}_{n_k}(f) - \bar{\mathbb{T}}_{n_k}(g) || < \frac{\epsilon}{3}, \forall k$. Since the set $\{f \in H^{1} \times H^{2} \times H^{1} \times H^{2}: || f || = 1\}$ is compact, it is totally bounded, i.e., $\exists\{f_j\}^{l}_{j=1}$, such that $H^{1} \times H^{2} \subset \bigcup^{l}_{j=1}B_{\delta}(f_{j})$. Since $\{\bar{\mathbb{T}}_{n_k}(f_{j})\}$ is strongly Cauchy for every $f_j$ (from Theorem \ref{Theorem 13}), given $\epsilon > 0, \exists N_{\epsilon} $ such that $|| \bar{\mathbb{T}}_{n_p}(f_{j}) - \bar{\mathbb{T}}_{n_q}(f_{j}) || < \frac{\epsilon}{3}, \forall p,q \geq N_{\epsilon}, j$.  For any $f \in H^{1} \times H^{2}$, $f \in B_{\delta}(f_{j}) $ for some $j$. Thus, for $p, q \geq N_{\epsilon}$, 
\begin{align*}
||  \bar{\mathbb{T}}_{n_p}(f) - \bar{\mathbb{T}}_{n_q}(f) || \leq  ||  \bar{\mathbb{T}}_{n_p}(f) - \bar{\mathbb{T}}_{n_p}(f_{j}) || + ||  \bar{\mathbb{T}}_{n_p}(f_j) - \bar{\mathbb{T}}_{n_q}(f_j) ||  + ||  \bar{\mathbb{T}}_{n_q}(f_{j}) - \bar{\mathbb{T}}_{n_q}(f) ||,
\end{align*} 
which is less than $\epsilon$. Thus, $||  \bar{\mathbb{T}}_{n_p} - \bar{\mathbb{T}}_{n_q} || \leq  \underset{f : || f ||= 1} \sup ||\bar{\mathbb{T}}_{n_p}(f) - \bar{\mathbb{T}}_{n_q}(f) || < \epsilon, \forall p, q, \geq N_{\epsilon}$. Thus, the sequence $\{ \bar{\mathbb{T}}_{n_k} \}$ of operators is Cauchy in space of continuous operators on $H^{1} \times H^{2}$ with the operator norm. Due to completeness of the space of continuous operators with the associated norm, the sequence converges to a continuous operator on $H^{1} \times H^{2}$.
\end{proof}
\begin{corollary}\label{Corollary 17}
There exists a fixed point for the operator $\mathbb{T}\big\vert_{H^{1} \times H^{2}}(\cdot)$ , i.e. $\exists f \in H^{1} \times H^{2}$, $|| f|| =1$, such that $\mathbb{T}(f) = f$.
\end{corollary}
\begin{proof}
$\mathbb{T}\big\vert_{H^{1} \times H^{2}}$ is continuous operator on $\{f \in H^{1} \times H^{2}: || f ||= 1\}$,  which is a compact and convex set. By \textit{Brouwer's fixed-point theorem}, the result follows.
\end{proof}
\section{Example}\label{Section 4}
\begin{figure}
\begin{center}
\includegraphics[scale=0.54]{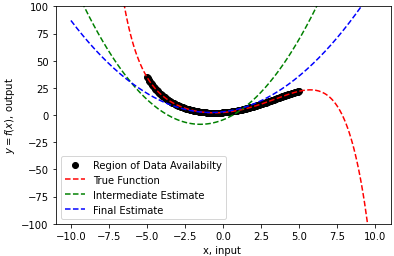}
\includegraphics[scale=0.54]{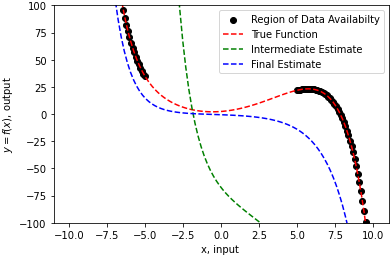}
\caption{True function and estimated functions at: (right) Agent 1 and (left) Agent 2.}\label{Figure 3}
\end{center}
\end{figure}
In this section, we consider a simple example to demonstrate the algorithm mentioned in Algorithm \ref{Algorithm 1}. True data is generated by considering a real valued function whose inputs are real values, obtained through a linear combination of polynomials and exponentials. Each agent receives noisy version of true data, i.e., noise added to the output data. Agent 1 considers the features, $\varphi^1_{1}(x) =1, \varphi^1_{2}(x) =x$ and $\varphi^1_3(x) = x^{2}$, while Agent 2 considers the features $\varphi^2_{1}(x) =\exp(-x), \varphi^2_{2}(x) =\exp(x)$. Thus, the kernel corresponding to Agent $1$ is $K^{1}(x,y) = 1 + xy + x^2 y^2$ and to Agent $2$ is $K^{2}(x,y) = \exp(-x -y) + \exp(x+y)$. The domain of the input data for Agent 1 is considered to be $[-5,5]$, while for Agent $2$ it is considered to be $[-10,-5] \cup [5,10]$. At time step $n$, after collecting data point $(x^{i}_n, y^{i}_n)$, each agent solves $(P1)^{i}_n$ to obtain $f^{i}_{n}$ which is then uploaded to the fusion space. 
\begin{wrapfigure}{r}{0.5\textwidth}
\begin{center}
\includegraphics[scale=0.54]{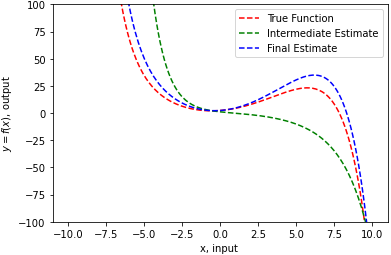}
\caption{True function and estimated functions at Fusion Center.}\label{Figure 4}
\vspace{-1.25cm}
\end{center}
\end{wrapfigure}
The fusion space corresponds to the RKHS generated by the kernel, $K(x,y) = 1 + xy + x^2 y^2 +  \exp(-x -y) + \exp(x+y)$. We note that the feature vectors of Agent $i$ are linearly independent. Thus, $H$ has dimension $5$. We chose $\bar{x}^{1}_{1} = 0, \bar{x}^{1}_{2}= 2, \bar{x}^{1}_{3} = 4, \bar{x}^{1}_{4} = -2, \bar{x}^{1}_{5} =-4$ and $\bar{x}^{2}_{1} = 1, \bar{x}^{1}_{2}= 3, \bar{x}^{1}_{3} = 5, \bar{x}^{1}_{4} = -1, \bar{x}^{1}_{5} =-3$. After estimating $\mathbf{\hat{Y}}^{i}_{n}$ as in Proposition \ref{Proposition 5}, the fusion problem in $(P3)_{n}$ is solved. To find the download operator, first, we choose the set of basis vectors for the space $H$ as $\varphi_1(x)=1, \varphi_2(x) = x, \varphi_3(x) = \sqrt{2}x^2 ,\varphi_4(x) =\exp(-x), \varphi_{5}(x) =\exp(x)$. With these basis vectors, the coefficients for $K^{1}(\cdot,y)$ are $[1;y;y^2; 0; 0]$, and for $K^{2}(\cdot,y)$ are $[0,0,0,\exp(-x),\exp(x)]$. Thus, the matrix representation for $\bar{L}^{i}$ is obtained as follows:
\begin{align*}
&\bar{L}^{1}(\varphi_1)(y)  \hspace{-2pt} = \hspace{-2pt} \langle \varphi_1(\cdot), K^{1}(\cdot,y) \rangle_{H} \hspace{-2pt} =\hspace{-3pt} \langle [1,0,0,0,0], [1,y, y^2, 0, 0] \rangle_{\mathbb{R}^5} =1. \;\hspace{-2pt}\bar{L}^{1}(\varphi_2)(y) \hspace{-2pt}  = \hspace{-2pt} \langle [0,1,0,0,0],\\ 
&[1,y, y^2, 0,0] \rangle_{\mathbb{R}^5} \hspace{-2pt} = \hspace{-2pt}y. \; \hspace{-2pt}\bar{L}^{1}(\varphi_3)(y) = \langle [0,0,1,0,0], [1,y, y^{2}, 0,0] \rangle_{\mathbb{R}^5} =y^{2}.\; \bar{L}^{1}(\varphi_4)(y) = 0.\\
&  \bar{L}^{1}(\varphi_5)(y) = 0, L^{1}_{M}= \sqrt{L^{1}_{M}}= 
\begin{bmatrix}
1 & 0 & 0 & 0 &0\\
0 & 1 & 0 & 0 &0\\
0 & 0 & 1 & 0 &0\\
0 & 0 & 0 & 0 &0\\
0 & 0 & 0 & 0 &0\\
\end{bmatrix}, \; 
L^{2}_{M}= \sqrt{L^{2}_{M}}= 
\begin{bmatrix}
0 & 0 & 0 & 0 &0\\
0 & 0 & 0 & 0 &0\\
0 & 0 & 0 & 0 &0\\
0 & 0 & 0 & 1 &0\\
0 & 0 & 0 & 0 &1\\
\end{bmatrix}.
\end{align*}
With this setup, simulations were run and the results are demonstrated in Figure \ref{Figure 3} and Figure \ref{Figure 4}. In each of the figures, the true function, function estimated at an intermediate iteration,  and the final estimate are plotted. In the figures pertaining to the agents, the segment of the true function which is accessible to the agents for data collection is also marked.  Our observations from the figures are as follows. Since the kernel for Agent 1 quadratic, the curve estimated by it is quadratic. The final estimate of Agent 1 partially overlaps with the true curve. Since the feature maps for Agent 2 are exponentials, the same is reflected in its estimates. However, neither of them are able to capture the true function ``completely". The final estimate in the fusion space captures the true function with minimum norm of the error, i.e., with minimum $|| f_{n^*} - f*||$ among all iterations, where $n^*$ is the final iteration. In reality, $f^*$ is not unknown and it would not be possible to verify the same. 
\section{Conclusion and Future Work}\label{Section 5}
To conclude, we presented a distributed algorithm for estimation of functions given data. Key aspects of the algorithm included use of heterogeneous data, use of different features by different agents, fusion of models by estimating data which could have generated the models, the use of uploading and downloading operators due to different members learning in different spaces, and, the consistency of the learning algorithm. Going forward, we are interested in studying nonparametric estimation problems using the algorithm developed in this paper. We would like to investigate transform methods for the fusion problem or the meta learning problems in RKHS. Quantification and formal guarantees of transfer of knowledge from one agent to another is also of interest. As mentioned in Section \ref{Section 1}, the formal study of knowledge is crucial for AI. Similarly, integrating reasoning with statistical learning is also crucial for enhanced learning. We are interested in understanding as to how reasoning can be integrated into a distributed learning algorithm mentioned in this paper.  
\section{Appendix}\label{Section 6}
\subsection{Fusion Space, Uploading and Downloading Operator}\label{Subsection 4.1}
The construction of the fusion space has been discussed in detail in \cite{raghavan2024distributed}. We mention the key results here for completion of the paper.
\begin{theorem}\label{Theorem 18}
If $K^{i}(\cdot,\cdot)$ is the reproducing kernel of Hilbert space $H^{i}$,  with norm $||\cdot||_{H^i}$, then $K(x,y)=K^{1}(x,y) + K^{2}(x,y)$ is the reproducing kernel of the space $H = \{f| f= f^1 + f^2 | f^{i} \in H^{i}\}$ with the norm:
\begin{align*}
||f||^2_{H} = \underset{\substack{f^1 + f^2 = f,\\ f^{i} \in H^{i}} }  \min \;\;  ||f^{1}||^2_{H^1} + ||f^{2}||^2_{H^2}.
\end{align*} 
\end{theorem}  
\begin{corollary}\label{Corollary 19}
The uploading operator from agent $i$'s knowledge space, $H^i$,  to the fusion space $H$, $\hat{L}^{i}: H^{i} \to H$, is $\hat{L}(f) = f$. $\hat{L}^{i}(\cdot)$, is linear and is bounded, $|| \hat{L}^{i} || = \sup\{ || f||_{H} : f\in H^{i}, || f ||_{H^{i}} =1 \} = 1$.
\end{corollary}
Let $H^{1} \times H^2$ be the product space with inner product $\langle (f^{1},f^{2}),(g^{1},g^{2})\rangle_{H^{1} \times H^2} = \langle f^{1}, g^{1}\rangle_{H^1} + \langle f^{2}, g^{2}\rangle_{H^2}$. Let $L: H^{1} \times H^2 \to H$ be a operator defined as $L((f^{1},f^{2})) =f^{1} +f^{2}$. $L$ is a linear operator and its null space, $\mathcal{N}(L) =\{(f^{1}, f^{2}) \in H^{1} \times H^2 : f^{1} +f^{2} =\theta \}$ is a closed subspace as it is finite dimensional. Thus, there exists a unique closed subspace $\mathcal{M}$ such that $H^{1} \times H^2 = \mathcal{M}\oplus \mathcal{N}(L)$, where $\mathcal{M} = \mathcal{N}^{\perp}$. The mapping $L_{\mathcal{M}} = L \circ \Pi_{\mathcal{M}}$ (operator $L$ restricted to subspace $\mathcal{M}$) is an isomorphism from $\mathcal{M}$ to $H$. Let the dimension of $\mathcal{M}$ be $m$ and the dimension of $\mathcal{N}(L)$ be $|\mathcal{I}^1| + |\mathcal{I}^2| -m$.
\begin{proposition}\label{Proposition 20}
Let $V= V^{1} \oplus V^2$ and $U = U^{1} \oplus U^2$ be Hilbert spaces such that $V \equiv U$ with $V^1=V^{2^{\perp}}$ and $U^1=U^{2^{\perp}}$. If $V^2$ is isomorphic to $U^{2}$ under the same isomorphism from $V$ to $U$, then $V^{1} \equiv U^{1}$. 
\end{proposition}
\begin{proof}
Since $V$ is isomorphic to $U$, there exits $\mathbb{L}: V \to U $ such that $\mathbb{L}$ is a bijection, is linear, and, has a well defined inverse $\mathbb{L}^{-1}: U \to V $ which is also linear. The subspace $V^{2}$ of $V$ is isomorphic to some subspace of $U$, however it is given that it is isomorphic to $U^2$. Thus $\mathbb{L}(V^{2}) = U^{2}$ and $\mathbb{L}^{-1}(U^{2}) = V^{2}$. Let $v \in V^{1}$, i.e, $\Pi_{V^{2}}(v) = \theta_{V}$. We claim that $\mathbb{L}(v) \in U^{1}$. Suppose not. Then, $\Pi_{U^{2}}(\mathbb{L}(v)) \neq \theta_{U}$, which implies that $ \mathbb{L}^{-1} (\mathbb{L}(v)) = \mathbb{L}^{-1}(\Pi_{U^{1}}(\mathbb{L}(v)) + \Pi_{U^{2}}(\mathbb{L}(v)) ) =  \mathbb{L}^{-1}(\Pi_{U^{1}}(\mathbb{L}(v))) + \mathbb{L}^{-1}(\Pi_{U^{2}}(\mathbb{L}(v)))$. Since $\mathbb{L}$ is an isomorphism from $V^2$ to $U^2$, $\mathbb{L}^{-1}(\Pi_{U^{2}}(\mathbb{L}(v))) \in V^{2} \neq \theta_{V}$. Hence, $\Pi_{V^{2}}( \mathbb{L}^{-1}(\Pi_{U^{2}}(\mathbb{L}(v))))  = \mathbb{L}^{-1}(\Pi_{U^{2}}(\mathbb{L}(v)))$. This implies that,
\begin{align*}
\Pi_{V^{2}} (v) = \Pi_{V^{2}}(\mathbb{L}^{-1} (\mathbb{L}(v))) = \Pi_{V^{2}}( \mathbb{L}^{-1}(\Pi_{U^{1}}(\mathbb{L}(v)))) + \mathbb{L}^{-1}(\Pi_{U^{2}}(\mathbb{L}(v))).
\end{align*}
Since there exists a unique $u' \in U^{2}$ such that $\mathbb{L}^{-1}(u') + \mathbb{L}^{-1}(\Pi_{U^{2}}(\mathbb{L}(v))) = \theta_V$, and,  $ \mathbb{L}^{-1}(\Pi_{U^{1}}(\mathbb{L}(v \\ )))$ cannot be written of the form $u' + u''$ as $U^{1} \cap U^{2} =\theta_{U}$, it follows that $\Pi_{V^{2}}( \mathbb{L}^{-1}(\Pi_{U^{1}}(\mathbb{L}(v)))) + \mathbb{L}^{-1}(\Pi_{U^{2}}(\mathbb{L}(v))) \neq \theta_V$  which is clearly a contradiction. Thus, $\mathbb{L}$ maps every $v \in V^{1}$ to a unique $u \in U^{1}$. Similarly, $\mathbb{L}^{-1}$ maps every $u \in U^{1}$ to unique $v \in V^{1}$. Since $\mathbb{L}$ and $\mathbb{L}^{-1}$ are linear, $V^{1}$ and $U^{1}$ are isomorphic under the morphism $\mathbb{L}$.
\end{proof}
The basis vectors, $\{\varphi^{1}_{j} \times \theta^2 \}_{j \in \mathcal{I}^{1}} \cup \{\theta^1 \times \varphi^{2}_{j} \}_{j \in \mathcal{I}^{2}}$, for $H^{1} \times H^2$ induce a isomorphism from $H^{1} \times H^2$ to $\mathbb{R}^{|\mathcal{I}^{1}| + |\mathcal{I}^2|}$. Under this isomorphism, $\mathcal{N}(L)$ is isomorphic to 
\begin{align*}
\mathcal{N}=\Big\{\Big(\boldsymbol{\alpha^1}, \boldsymbol{\alpha^2} \Big) \in \mathbb{R}^{|\mathcal{I}^1| + |\mathcal{I}^2|}: \sum_{j \in \mathcal{I}^1}\alpha^1_j \varphi^{1}_{j} + \sum_{j \in \mathcal{I}^2}\alpha^2_j \varphi^{2}_{j}= \theta,
\boldsymbol{\alpha^i} = \Big(\alpha^{i}_{1}, \ldots, \alpha^{i}_{|\mathcal{I}^i|}\Big), i=1,2\Big\}.
\end{align*}
\begin{proposition} \label{Proposition 21}
There exists two sets of $m$ elements each, $\{\bar{x}^{i}_{j}\}^{m}_{j=1}$, $i=1,2$ such that (i) $\{\bar{x}^{1}_{j}\}^{m}_{j=1} \cap \{\bar{x}^{2}_{j}\}^{m}_{j=1} = \emptyset$; (ii) $\{K(\cdot,\bar{x}^{1}_{j})\}^{m}_{j=1}$  and $\{K(\cdot,\bar{x}^{2}_{j})\}^{m}_{j=1}$, each form a basis for $H$. Thus, for any function $f \in H$, $\exists! \{\alpha^{i}_{j}\}^{m}_{j=1}$ such that $f(\cdot) = \sum^{m}_{j=1} \alpha^{1}_{j}K(\cdot,\bar{x}^{1}_{j}) = \sum^{m}_{j=1}\alpha^{2}_{j}K(\cdot,\bar{x}^{2}_{j})$.
\end{proposition}
\begin{proof}
Let $\mathbb{R}^{|\mathcal{I}^1| + |\mathcal{I}^2|} = \bar{\mathcal{M}} \oplus \mathcal{N}$, where,  $\bar{\mathcal{M}} = \mathcal{N}^{\perp}$. From Proposition \ref{Proposition 20}, under the isomorphism induced by the basis vectors of $H^{1} \times H^2$, $\mathcal{M}$ and $\bar{\mathcal{M}}$ are isomorphic. Let, $\boldsymbol{\varphi}(x) = [\varphi^{1}_{1}(x), \ldots, \varphi^{1}_{\mathcal{I}^1}(x),\;  \varphi^{2}_{1}(x), \ldots, \varphi^{2}_{\mathcal{I}^1}(x)] \in  \mathbb{R}^{|\mathcal{I}^1| + |\mathcal{I}^2|}, \;x \in \mathcal{X}$. Let $\hat{\mathcal{M}}$ be the span of $\{\boldsymbol{\varphi}(x)\}_{x \in \mathcal{X}}$. From the definition of $\mathcal{N}$, it follows that $\mathcal{N} = \hat{\mathcal{M}}^{\perp}$. Since $\mathcal{N}$ is a closed subspace, $\mathcal{N} = \mathcal{N}^{\perp ^{\perp}} = \bar{\mathcal{M}}^{\perp}$. Thus, $\hat{\mathcal{M}}^{\perp} = \bar{\mathcal{M}}^{\perp}$ which implies that $\hat{\mathcal{M}}^{\perp^{\perp}} = \bar{\mathcal{M}}^{\perp^{\perp}}$. Since $\hat{\mathcal{M}}$ and $\bar{\mathcal{M}}$ are finite dimensional subspaces, they are closed, which implies that $\hat{\mathcal{M}} = \hat{\mathcal{M}}^{\perp^{\perp}} = \bar{\mathcal{M}}^{\perp^{\perp}} =  \bar{\mathcal{M}}$. Hence, $\hat{\mathcal{M}} = \bar{\mathcal{M}}$. Since $H$ is isomorphic to $\mathcal{M}$, it is isomorphic to $\bar{\mathcal{M}}$ and hence to $\hat{\mathcal{M}}$. We choose a basis for $\hat{\mathcal{M}}$ as follows. First, we choose $\bar{x}^1_{1}$ arbitrarily to obtain the first basis vector $\boldsymbol{\varphi}(\bar{x}^1_{1})$. Let $\hat{M} = \hat{M}^{1}_{1} \oplus \text{Span}\Big(\boldsymbol{\varphi}(\bar{x}^1_{1})\Big)$. $\boldsymbol{\varphi}(\bar{x}^1_{2})$ is chosen from $\hat{M}^{1}_{1}$. This process is repeatedly iteratively where $\boldsymbol{\varphi}(\bar{x}^1_{j})$ is chosen from $\hat{\mathcal{M}}^{1}_{j-1}$ with $\hat{\mathcal{M}} = \hat{\mathcal{M}}^{1}_{j-1} \oplus \text{Span}\Big(\{\boldsymbol{\varphi}(\bar{x}^1_{k})\}^{k=j-1}_{k=1}\Big)$ for $j=2, \ldots, m$. Thus, $\{\boldsymbol{\varphi}(\bar{x}^1_{j})\}^{m}_{j=1}$ spans $\hat{\mathcal{M}}$. Using the isomorphism induced by the basis vectors of $H^{1} \times H^2$ from $\hat{\mathcal{M}}$ to $\mathcal{M}$, we note that each $\boldsymbol{\varphi}(\bar{x}^1_{k})$ gets mapped to $\Big( \sum_{j \in  \mathcal{I}^{1}}\varphi^1_{j}(\bar{x}^1_{k})(\varphi^1_{j}(\cdot) \times \theta^{2}) + \sum_{j \in \mathcal{I}^{2}} \varphi^2_{j}(\bar{x}^1_{k})( \theta^{1} \times \varphi^2_{j}(\cdot)) \Big) = \Big( \sum_{j \in  \mathcal{I}^{1}}\varphi^1_{j}(\bar{x}^1_{k})\varphi^1_{j}(\cdot), \sum_{j \in \mathcal{I}^{2}} \varphi^2_{j}(\bar{x}^1_{k})\varphi^2_{j}(\cdot) \Big)$. Invoking the isomorphism $L_{\mathcal{M}}$ from $\mathcal{M}$ to $H$, each $\boldsymbol{\varphi}(\bar{x}^1_{k})$ gets mapped to $\sum_{j \in  \mathcal{I}^{1}}\varphi^1_{j}(\bar{x}^1_{k})\varphi^1_{j}(\cdot) + \sum_{j \in \mathcal{I}^{2}} \varphi^2_{j}(\bar{x}^1_{k})\varphi^2_{j}(\cdot) = K(\cdot,\bar{x}^{1}_{k})$. Thus, $\{ K(\cdot,\bar{x}^{1}_{j})\}^{m}_{j=1}$ spans $H$. Let $\tilde{\mathcal{M}}$ be the span of $\{\boldsymbol{\varphi}(x)\}_{x \in \mathcal{X} \sim \{\bar{x}^{1}_{j}\}^{m}_{j=1}}$ and $\tilde{\mathcal{N}}$ be defined as 
\begin{align*}
\tilde{\mathcal{N}}=\Big\{\Big(\boldsymbol{\alpha^1}, \boldsymbol{\alpha^2} \Big) \in \mathbb{R}^{|\mathcal{I}^1| + |\mathcal{I}^2|}: \sum_{j \in \mathcal{I}^1}\alpha^1_j \varphi^{1}_{j}(x) + \sum_{j \in \mathcal{I}^2}\alpha^2_j \varphi^{2}_{j}(x) &= 0, \forall x \in  \mathcal{X} \sim \{\bar{x}^{1}_{j}\}^{m}_{j=1},\\
&\boldsymbol{\alpha^i} = \Big(\alpha^{i}_{1}, \ldots, \alpha^{i}_{|\mathcal{I}^i|}\Big)\Big\},  i=1,2.
\end{align*}
Clearly, $\mathcal{N} \subset \tilde{\mathcal{N}}$. Suppose $\mathcal{X}$ is closed, connected subset of $\mathbb{R}^d$ without isolated points. Let $\{x_{n}\} \subset \mathcal{X} \sim \{\bar{x}^{1}_{j}\}^{m}_{j=1}$ be a sequence such that it converges to one of the $\bar{x}^{1}_{j}$. Suppose $\Big(\boldsymbol{\alpha^1}, \boldsymbol{\alpha^2} \Big)  \in \tilde{\mathcal{N}}$. Then, 
\begin{align*}
\sum_{j \in \mathcal{I}^1}\alpha^1_j \varphi^{1}_{j}(x_n) + \sum_{j \in \mathcal{I}^2}\alpha^2_j \varphi^{2}_{j}(x_n) = 0,  \forall n \implies \underset{n \to \infty}\lim \sum_{j \in \mathcal{I}^1}\alpha^1_j \varphi^{1}_{j}(x_n) + \sum_{j \in \mathcal{I}^2}\alpha^2_j \varphi^{2}_{j}(x_n) = 0
\end{align*}
By continuity of the feature maps the above implies, $\sum_{j \in \mathcal{I}^1}\alpha^1_j \varphi^{1}_{j}(\bar{x}^{1}_{j}) + \sum_{j \in \mathcal{I}^2}\alpha^2_j \varphi^{2}_{j}(\bar{x}^{1}_{j}) = 0$. Since the same argument can be presented for all  $\{ \bar{x}^{1}_{j}\}^{m}_{j=1}$, this implies that $\Big(\boldsymbol{\alpha^1}, \boldsymbol{\alpha^2} \Big)  \in \mathcal{N}$. Thus, $\mathcal{N} = \tilde{\mathcal{N}}$ and $\hat{\mathcal{M}} =\tilde{\mathcal{M}}$. Since there is no loss in dimensionality, $\{\boldsymbol{\varphi}(\bar{x}^{2}_{j})\}^{m}_{j=1}$ can be chosen using exactly the same procedure described to choose $\{\boldsymbol{\varphi}(\bar{x}^{1}_{j})\}^{m}_{j=1}$, however using $\tilde{\mathcal{M}}$ in the place of $\hat{\mathcal{M}}$ leading to the construction of $\{K(\cdot,\bar{x}^{2}_{j})\}^{m}_{j=1}$ which spans $H$.
\end{proof}
We recall the following results for continuity and completeness of the paper. 
\begin{lemma}\label{Lemma 22}
Given the RKHS, $(H, \langle \cdot,\cdot \rangle_{H})$, with kernel $K(\cdot,\cdot)$ and the kernels $K^{i}(\cdot,\cdot),\;i=1,2$, such that $K(x,y) = K^{1}(x,y) + K^{2}(x,y)$, we define operators, $\bar{L}^{i}: H \to H$, as
\begin{align*}
\bar{L}^{i}(f)(x) =\langle f(\cdot), K^{i}(\cdot,x) \rangle_{H}, \text{ for}, i=1,2.
\end{align*}
Then, $\bar{L}^{i}$ is linear, symmetric, positive and bounded, $|| \bar{L}^{i} ||\leq 1$.
\end{lemma}
\begin{theorem}\label{Theorem 23}
The linear space  $\bar{H}^{i} = \{g : g = \sqrt{\bar{L}^i}(f), f\in H \}$ is a RKHS with kernel $K^i$.  $\sqrt{\bar{L}^i}(\cdot)$ establishes an isometric isomorphism between $\mathcal{N}\big(\sqrt{\bar{L}^i}\big)^{\perp}$ and $\bar{H}^{i}$, and the norm, $||f||_{\bar{H}^{i}} = ||g||_{H}$,  where $f = \sqrt{\bar{L}^i} g,  g\in \mathcal{N}\big(\sqrt{\bar{L}^i}\big)^{\perp}$. The downloading operator from the fusion space $H$ to  agent $i$'s knowledge space, $H^i$, is $\sqrt{\bar{L}^{i}} \circ \Pi_{\mathcal{N}\big(\sqrt{\bar{L}^i}\big)^{\perp}}$. The downloading operator is linear and bounded.
\end{theorem}
\begin{proposition}\label{Proposition 24}
The uploaded and downloaded function at Agent $i$ can be expressed uniquely as a linear combination of $\{K^{i}(\cdot, \bar{x}^{i}_j)\}^{m}_{j=1}$.
\end{proposition}
\begin{proof}
Since $\bar{L}^{i}$ is symmetric, from the spectral theorem, it follows that (i) the eigenvectors of $\bar{L}^{i}$, $\{\bar{\varphi}^{i}_j\}^{m}_{j=1}$  are an orthonormal basis for $H$; (ii) the eigenvalues of $\bar{L}^{i}$, $\{\lambda^{i}_{j}\}^{m}_{j=1}$, are real. The square root of operator $\bar{L}^{i}$, is defined as $\sqrt{\bar{L}^i}\big(\bar{\varphi}^{i}_j\big) = \sqrt{\lambda^{i}_{j}}\bar{\varphi}^{i}_j$. Any $f \in H$, specifically the fused function at stage $n$, $f_{n}$ is first expressed using the eigen vectors of $\bar{L}^{i}$ as $f_{n} = \sum^{m}_{k=1}b^{i}_{n,k}\bar{\varphi}^{i}_k$. Invoking Proposition \ref{Proposition 21}, $\bar{\varphi}^{i}_k = \sum^{m}_{j=1}a^{i}_{k,j}K(\cdot, \bar{x}^{i}_{j})$. This implies that, $\bar{L}^{i}\big(\bar{\varphi}^{i}_k\big) = \sum^{m}_{j=1}a^{i}_{k,j}\bar{L}^{i}\big(K(\cdot, \bar{x}^{i}_{j})\big) = \sum^{m}_{j=1}a^{i}_{k,j}K^{i}(\cdot, \bar{x}^{i}_{j})$. Since $\bar{L}^{i}(\bar{\varphi}^{i}_k) = \lambda^{i}_{k}\bar{\varphi}^{i}_k$, it follows that $\bar{\varphi}^{i}_k = \frac{1}{\lambda^i_k}\sum^{m}_{j=1}a^{i}_{k,j}K^{i}(\cdot, \bar{x}^{i}_{j})$, $\lambda^{i}_{k} \neq 0$. From Theorem \ref{Theorem 23}, we note that, $\{\bar{\varphi}^{i}_j\}_{\lambda^{i}_{j} \neq 0}$ spans $H^{i}$. Hence, any vector $f^{i} \in H^i$, $f^{i} =  \underset{k: \lambda^{i}_{j} \neq 0}\sum c^{i}_{k}\sum^{m}_{j=1}a^{i}_{k,j}K^{i}(\cdot, \bar{x}^{i}_{j})$, which implies that $\{K^{i}(\cdot, \bar{x}^{i}_{j})\}^m_{j=1}$ spans $H^{i}$. Thus, the uploaded function can be expressed uniquely as a linear combination of  $\{K^{i}(\cdot, \bar{x}^{i}_j)\}^{m}_{j=1}$.\\
For $\lambda^{i}_{k} \neq 0$, $\sqrt{\bar{L}^i}\big(\bar{\varphi}^{i}_k\big) = \frac{1}{\sqrt{\lambda^{i}_{k}}}\bar{L}^{i}(\bar{\varphi}^{i}_k)  =  \frac{1}{\sqrt{\lambda^{i}_{k}}} \Big(  \sum^{m}_{j=1}a^{i}_{k,j}K^{i}(\cdot, \bar{x}^{i}_{j}) \Big)$. From Theorem \ref{Theorem 23}, the function downloaded on to the knowledge space of the agent $i$ at stage $n$, $\bar{f}^{i}_{n}$, is 
\begin{align*}
\sqrt{\bar{L}^{i}} \circ \Pi_{\mathcal{N}\big(\sqrt{\bar{L}^i}\big)^{\perp}} (f_{n}) = \underset{k: \lambda^{i}_{k} \neq 0} \sum b^{i}_{n,k} \sqrt{\bar{L}^i}\big(\bar{\varphi}^{i}_k\big)  &= \underset{k: \lambda^{i}_{k} \neq 0} \sum \frac{b^{i}_{n,k}}{\sqrt{\lambda^{i}_{k}}} \Big(  \sum^{m}_{j=1}a^{i}_{k,j}K^{i}(\cdot, \bar{x}^{i}_{j}) \Big)  \\
&= \sum^{m}_{j=1} \Big(  \underset{k: \lambda^{i}_{k} \neq 0} \sum   \frac{b^{i}_{n,k} a^{i}_{k,j}}{\sqrt{\lambda^{i}_{k}}}\Big) K^{i}(\cdot, \bar{x}^{i}_{j})
\end{align*}
Thus, downloading the fused function is equivalent to the fusion center transmitting the vector $\Big\{\underset{k: \lambda^{i}_{k} \neq 0} \sum   \frac{b^{i}_{n,k} a^{i}_{k,j}}{\sqrt{\lambda^{i}_{k}}}\Big\}^{m}_{j=1}$ to agent $i$. 
\end{proof}
To obtain the closed form expression for the downloaded function, we note that it is crucial to express the fused function in terms of the eigenvectors of $\bar{L}^{i}$. If not, if $f_{n}$ is expressed directly using the basis vectors $\{K(\cdot,\bar{x}^{i}_{j})\}^{m}_{j=1}$, then computation of $\sqrt{\bar{L^{i}}} \big(K(\cdot, \bar{x}^{i}_{j})\big)$ is not straight forward and $\sqrt{\bar{L}^{i}}(f) = \frac{1}{\sqrt{\lambda}}\bar{L}^{i}(f)$ if only if $f$ is an eigenvector of $\bar{L}^{i}$. 
\acks{Research supported by the Swedish Research Council (VR), Swedish Foundation for Strategic Research (SSF),  and the Knut and Alice Wallenberg Foundation.}
\bibliography{biblio}
\end{document}